\documentclass[acmsmall,screen,nonacm]{acmart}

\usepackage{cleveref}
\usepackage{cancel}
\usepackage{mathtools}
\usepackage{mathpartir}
\usepackage{stmaryrd}
\usepackage{listings}
\usepackage{xcolor}
\usepackage{enumitem}
\usepackage{booktabs}

\crefformat{section}{\S#2#1#3}
\crefmultiformat{section}{\S#2#1#3}{ and \S#2#1#3}{, #2#1#3}{ and \S#2#1#3}
\crefformat{subsection}{\S#2#1#3}
\crefmultiformat{subsection}{\S#2#1#3}{ and \S#2#1#3}{, #2#1#3}{ and \S#2#1#3}
\crefformat{subsubsection}{\S#2#1#3}
\crefmultiformat{subsubsection}{\S#2#1#3}{ and \S#2#1#3}{, #2#1#3}{ and \S#2#1#3}

\definecolor{kwcolor}{RGB}{0,0,160}
\definecolor{commentcolor}{RGB}{101,123,131}
\definecolor{strcolor}{RGB}{42,120,118}
\definecolor{paramcolor}{RGB}{181,137,0}
\definecolor{opcolor}{RGB}{165,40,110}

\lstdefinelanguage{cambria}{
  morekeywords={with,handle,handler,return,do,in,fun,rec,if,then,else,case,of,inl,inr,effect,finally},
  morekeywords=[2]{\$p,\$loc,\$name,\$tid,\$node,\$urn},
  morekeywords=[3]{!act,!bernoulli,!eq,!fork,!fresh,!get,!newurn,!node,!op,!print,!ref,!sample,!set,!stop,!uniform,!wait},
  sensitive=true,
  alsoletter={\$!},
  morecomment=[l]{--},
  morestring=[b]",
  literate=
    {->}{$\to$}2
    {<-}{$\leftarrow$}2
    {~>}{$\leadsto$}2
    {=>}{$\Rightarrow$}2
    {++}{$\mathbin{+\!\!+}$}2
    {forall}{$\forall$}1,
}

\makeatletter
\newcommand{\cul}[1]{\mathpalette\cul@{#1}}
\newcommand{\cul@}[2]{\sbox\z@{$\m@th#1#2$}%
  \mathord{\copy\z@\kern-\wd\z@\rule[-1.6pt]{\wd\z@}{0.5pt}}}
\makeatother
\newcommand{\cC}{\cul{C}}
\newcommand{\cD}{\cul{D}}

\newcommand{\cambria}{\textsc{Cambria}}

\newcommand{\mlstinline}[1]{\mbox{\lstinline|#1|}}
\newcommand{\xstr}{\mlstinline{"x"}}

\newcommand{\void}{\mathsf{Void}}
\newcommand{\tunit}{\mathsf{Unit}}
\newcommand{\tint}{\mathsf{Int}}
\newcommand{\tbool}{\mathsf{Bool}}
\newcommand{\tdouble}{\mathsf{Double}}
\newcommand{\tstring}{\mathsf{Str}}
\newcommand{\tname}{\mathsf{Name}}

\newcommand{\tfun}[2]{#1 \to #2}

\newcommand{\thandler}[2]{#1 \Rightarrow #2}

\newcommand{\comp}[2]{#1\;!\;#2}
\newcommand{\arity}[2]{#1 \to #2}

\newcommand{\op}{\mathsf{op}}
\newcommand{\ret}{\mathsf{return}\;}
\newcommand{\ddo}{\mathsf{do}\;}
\newcommand{\iin}{\;\mathsf{in}\;}

\newcommand{\fun}{\mathsf{fun}\;}
\newcommand{\rec}{\mathsf{rec}\;}
\newcommand{\hhandler}{\mathsf{handler}}
\newcommand{\with}{\mathsf{with}\;}
\newcommand{\handle}{\;\mathsf{handle}\;}

\newcommand{\dom}{\mathsf{dom}}
\newcommand{\cod}{\mathsf{cod}}

\newcommand{\pack}{\mathsf{pack}}
\newcommand{\unpack}{\mathsf{unpack}\;}
\newcommand{\as}{\;\mathsf{as}\;}
\newcommand{\unit}{(\,)}

\newcommand{\rref}{\mathsf{ref}}
\newcommand{\sset}{\mathsf{set}}
\newcommand{\fresh}{\mathsf{fresh}}

\newcommand{\eq}{\mathsf{eq}}

\newcommand{\G}{\Gamma}
\newcommand{\Sig}{\Sigma}
\newcommand{\fpv}{\mathsf{fpv}}
\newcommand{\fv}{\mathsf{fv}}
\newcommand{\rel}{\mathsf{rel}}
\newcommand{\Beta}{\mathsf{Beta}}
\newcommand{\Uniform}{\mathsf{Uniform}}

\newcommand{\corestep}{\rightsquigarrow}

\newcommand{\nofin}[1]{\bar{#1}}

\newcommand{\tj}[3]{#1; #2 \vdash #3}
\newcommand{\vj}[5]{#1; #2; #3 \vdash #4 : #5}
\newcommand{\cj}[6]{#1; #2; #3 \vdash #4 : \comp{#5}{#6}}
\newcommand{\cjp}[5]{#1; #2; #3 \vdash #4 : #5}
\newcommand{\wj}[5]{#1; #2; #3 \vdash_{\mathsf w} #4 : #5}
\newcommand{\vdashc}{\vdash_{\mathsf c}}
\newcommand{\tjc}[3]{#1; #2 \vdashc #3}
\newcommand{\vjc}[5]{#1; #2; #3 \vdashc #4 : #5}

\newcommand{\cjpc}[5]{#1; #2; #3 \vdashc #4 : #5}
\newcommand{\eenv}{\bullet}

\newcommand{\interp}[2][\rho]{{\llbracket #2 \rrbracket_{#1}}}
\newcommand{\siminterp}[2][\rho]{{\llbracket #2 \rrbracket^{\approx}_{#1}}}
\newcommand{\Rel}{\mathbf{Rel}}

\newcommand{\ctx}{K}

\setcopyright{none}
\acmJournal{PACMPL}
\acmVolume{1}
\acmNumber{POPL}
\acmArticle{1}
\acmYear{2027}
\acmMonth{1}
\acmDOI{}

\begin{document}

\title{Cambria: Resource Abstraction for Parametrized Algebraic Effects and Handlers}

\author{Jack Liell-Cock}
\orcid{0009-0005-7121-8095}
\affiliation{%
  \institution{University of Oxford}
  \city{Oxford}
  \country{United Kingdom}
}
\email{jack.liell-cock@cs.ox.ac.uk}

\author{Sam Staton}
\orcid{0000-0002-7149-3805}
\affiliation{%
  \institution{University of Oxford}
  \city{Oxford}
  \country{United Kingdom}
}
\email{sam.staton@cs.ox.ac.uk}

\begin{abstract}
The algebraic effects and handlers paradigm separates the concerns of the interface and implementation of computational effects in programming languages.
We present \cambria{}, a language that extends this framework to the parametrized setting.
Effect signatures may use abstract parameter types that are instantiated by the handler along with the operation implementations.
Parameters abstract over resources, such as memory locations or thread IDs, permitting algebraic effects to encode dynamic allocation.
They are first-class in the type system but erased at runtime, requiring no coercions or type-directed reduction.

We prove parametricity via a step-indexed logical relation, formalizing the abstraction guarantee provided by parametrized handlers.
We also establish type safety and classify the annotations needed for completeness of the type inference algorithm.
We demonstrate \cambria{}'s practicality with a working implementation
and provide examples including local state, P\'olya's urn, and concurrent thread management.
The last is a parametrized effect whose abstract thread IDs are shared between concurrent computations, going beyond standard instances.
\cambria{} is the first calculus with user-defined resource-allocating effects that guarantees, via parametricity,
that client code cannot depend on how a handler represents its resources.
\end{abstract}



\maketitle

\section{Introduction}\label{sec:intro}

Algebraic effects and handlers~\cite{plotkin_handling_2009} provide a principled approach to programming with side effects.
A computation performs abstract \emph{operations} such as reading state or spawning a thread,
and a \emph{handler} gives these operations meaning by supplying an implementation.
This separation of interface from implementation has made algebraic effects an increasingly popular foundation for effect systems in both research and practice~%
\cite{bauer_programming_2015, hillerstrom2016liberating, leijen_koka_2017, lindley_do_2017, sivaramakrishnan_retrofitting_2021}.

An effect's interface often mentions types that the client \emph{uses} but should never \emph{inspect}.
The underlying representation is the handler's choice alone. Local state is the classic example.
Client code may allocate memory locations, and read and write through them, but it should never need to know how a location is represented.
For example, a location could be an index in a list, a unique name in a map, or an opaque token from the runtime.
Committing the client to a concrete representation up front buys it nothing and forces unnecessary bookkeeping.
Moreover, by keeping the locations abstract, the client cannot forge new ones, inspect them for equality, or perform arithmetic on them,
except through the operations the interface provides.
Thus, abstraction also provides a safety guarantee.

We present \cambria{}, a language that extends algebraic effects and handlers to the parametrized setting.
Formally, parametrized handlers are derived from the theory of parametrized algebraic theories (PATs)~\cite{Staton13, Staton13PL},
which extend classic algebraic theories with the ability of operators to bind and use abstract parameters in their arguments, denoting resource allocation and usage.
Then a model of the theory, or a handler, instantiates the parameter as a type, as well as interpreting the operations.
This generalization was introduced to handle the problem of \emph{instances of effects},
where a computational effect may dynamically allocate resources at runtime, such as new memory locations in local state.
Instances of effects are a leading case for parametrized handlers, but a parameter is more generally any abstract type that an operation introduces and a handler instantiates.
We demonstrate that parametrized handlers, as the operationalization of PATs, enable idiomatic implementations of diverse effects and parameters (\Cref{sec:local-state}--\ref{sec:concurrency}).

One approach to instances is to represent each as a concrete object in the program.
Eff originally provided first-class instances in this way~\cite{bauer_programming_2015}, but later removed them,
and subsequent calculi recovered instances as dynamically allocated labels~\cite{biernacki2020binders, devilhena2023labels}, first-class names~\cite{xie2022firstclass}, or capabilities~\cite{brachthauser2020capabilities}.
The drawback of these approaches is that either the resources are not first-class, or the representation is fixed and imposed on the client.
\cambria{}'s alternative approach to this problem separates the concerns of \emph{dispatch} and \emph{instances}.
That is, a single handler handles all instances of the same effect, with the operators encoding an interface to multiple instances.
For example, with local state, rather than having multiple copies of the \lstinline|!get| and \lstinline|!set| operators,
we have an additional \lstinline|!ref| operator that creates a new memory location that becomes an argument to \lstinline|!get| and \lstinline|!set|.

Effects such as dynamically threaded concurrency~\cite{kammar2025equational} and code jumps~\cite{staton2014jumps} go beyond independent instances.
Parameters are shared and used to coordinate different sub-computations.
This motivates a single handler interpreting all instances, rather than a witness distinguishing each one,
and \cambria{} provides this while keeping the parameters abstract.
In the concurrency example, the effect operations encode POSIX-like fork and wait.
Forking dynamically spawns a thread, the reference to which is an abstract thread ID that other processes can wait on.
The ability to instantiate the abstract thread ID in the parametrized handler means that each scheduler implementation remains self-contained and simple.
In \Cref{sec:concurrency}, we demonstrate this with three different handlers implementing child-biased scheduling, parent-biased scheduling, and dependency diagnostics.

The key method underlying \cambria{} is that parametrized handlers decompose into classic handlers and existential types.
This lets the surface language stay ergonomic while the metatheory reduces to standard, modular pieces.
The main result, parametricity, then gives representation independence for effects with dynamic resources without keeping track of allocation.
The decomposition also eases adoption since a language that already provides effect handlers and existential types, such as OCaml~5~\cite{sivaramakrishnan_retrofitting_2021},
can support parametrized handlers with little new core machinery.
The addition is moreover backwards compatible since a handler with no parameters is exactly a classic handler.

In particular, we make the following contributions:
\begin{enumerate}
  \item \textbf{Language design.} We present \cambria{}, a core calculus combining algebraic effect handlers, parameter variables, let-polymorphism,
    and general recursion, and establish type safety (\Cref{sec:overview}--\ref{sec:calculus}).

  \item \textbf{Parametricity.} We prove a parametricity theorem via a step-indexed logical relation (\Cref{sec:parametricity}).
    This rests on an elaboration result (\Cref{sec:elaboration}) that shows parametrized handlers decompose into classic handlers and existential types.

  \item \textbf{Type inference.} We establish an annotation completeness result (\Cref{sec:inference}) that states that
    it suffices to annotate only operations with arities mentioning abstract parameter types.
    All other types, including polymorphic ones, are inferred.

  \item \textbf{Expressiveness.} We demonstrate that \cambria{} can express parametrized effects with abstraction guarantees that prior work cannot,
    including concurrent thread management from recent work on dynamic threads~\cite{kammar2025equational}.
    We provide an implementation with runnable examples (\Cref{sec:overview}).
\end{enumerate}

All results novel to \cambria{} are mechanized in the Rocq prover, relying only on functional extensionality.

\section{Overview of the \cambria{} Language}\label{sec:overview}

We introduce \cambria{} through a series of examples that exercise its key features: parameter variables, let-polymorphism, general recursion, and their interactions.
Each example is executable in our implementation.

\subsection{Local State and Fresh Name Generation} \label{sec:local-state}

The canonical example of a dynamic effect is local state, where we have an additional operation \lstinline|!ref|
that allocates a reference given an initial value, along with the standard operations \lstinline|!get| and \lstinline|!set|
that read from and write to memory locations of type \lstinline|$loc|, where the leading dollar sign marks an abstract parameter type.
The following program uses this effect to allocate two reference cells and manipulate data between them before returning the value at one of the locations.
\begin{lstlisting}
effect !ref : Int  ~> $loc.
effect !get : $loc ~> Int.
effect !set : $loc * Int ~> Unit.
do a <- !ref 2 in
do b <- !ref 3 in
do _ <- !set (a, !get b) in
!get a
\end{lstlisting}
Note that here the \lstinline|effect| declarations serve merely to specify the operation arities, and in particular which types are abstract parameter types (\lstinline|$p|) rather than concrete types left to be inferred (\Cref{sec:inference}).
The \lstinline|effect| declaration allocates nothing in itself; all generative behaviour is up to the handler of an operation such as \lstinline|!ref|, as we see next.
(See \Cref{sec:related} for further comparison with other systems.)

The parametrized handle construct has an additional parameter substitution clause \lstinline|[$p -> A]|,
allowing the abstract parameter types in a computation \lstinline|c| to be instantiated before being handled by a (classic) handler \lstinline|h|.
\begin{lstlisting}
with [$p -> A] h handle c
\end{lstlisting}
Multiple parameter instantiations are separated by commas, and the instantiation clause can be dropped when there are none,
recovering the classic handle construct.
In the examples that follow, we show each handler with the instantiation it is applied under,
writing \lstinline|[$p -> A] handler {...}| for the handler and substitution pair used in \lstinline|with [$p -> A] h handle c|.

We may handle the local state effect using the classic parameter passing trick for stateful effects.
Rather than a global value, we choose the state to be a function from the names of our memory locations to the values that they store.
The function encoding is for self-containment; any associative data structure would work equally well.
The convenience of parametrized effect handlers is that we need not worry about the implementation of the memory location \lstinline|$name| type yet;
we just need an interface for fresh name generation \lstinline|!fresh : Unit ~> $name|, and equality checking \lstinline|!eq : $name * $name ~> Bool|.
Then, we can use another handler to implement this interface at a later stage.
\begin{lstlisting}
[$loc -> $name] handler {
  return x    -> return (fun _ -> return x),
  get a k     -> return (fun s -> k (s a) s),
  set (a,x) k -> return (fun s -> k () (fun b ->
                   if !eq (a, b) then return x else s b)),
  ref x k     -> do a <- !fresh () in return (fun s ->
                   k a (fun b -> if !eq (a, b) then return x else s b)),
  finally s   -> s (fun _ -> return 0)
}
\end{lstlisting}
This function encoding necessitates a default value for initialization (in this case 0), but parametricity (\Cref{sec:parametricity})
guarantees that only locations initialized with \lstinline|!ref| can be accessed, making this default value observationally meaningless.
A genuine drawback of the encoding is that it has no means of deallocating a reference that has gone out of scope.
We argue in \Cref{sec:equations} that client-transparent liveness analysis for reachability-based garbage collection would violate parametricity,
and discuss alternatives.

The fresh name effect could be handled by parameter passing and incrementing an integer,
concretizing the \lstinline|$loc| and \lstinline|$name| types to integers.
\begin{lstlisting}
[$name -> Int] handler {
  return x   -> return (fun _ -> return x),
  fresh _ k  -> return (fun n -> k n (n+1)),
  eq (a,b) k -> k (a == b),
  finally f  -> f 0
}
\end{lstlisting}
Despite the concretization, the following code handled in the same way would not type-check,
because \lstinline|a| is parametric from the perspective of the client code.
\begin{lstlisting}
do a <- !ref 2 in !set (a+1, 3)
\end{lstlisting}
The following sections demonstrate that the programmer's freedom to choose the interpretation of parameters
on top of the operations leads to modular and expressive language constructs.

\subsubsection{Polymorphism Meets Abstraction}\label{sec:poly-refs}

Parametrized handlers support statically unbounded creation of effect instances.
The number of memory locations can vary based on runtime data.
Consider the following \cambria{} program that uses a polymorphic \lstinline|map| to allocate a memory location for each integer in a list,
then uses the same function to access each location.
The primitives \lstinline|[]| and \lstinline|(::)| are the constructors for the empty list and element prepending,
while \lstinline|uncons : List a -> Unit + a * List a| unrolls the list constructor once.
\begin{lstlisting}
effect !ref : Int ~> $loc.
effect !get : $loc ~> Int.
effect !set : $loc * Int ~> Unit.
do map <- return (rec map f xs ->
  case uncons xs of {
    inl _        -> return [],
    inr (x, xs') -> f x :: map f xs'
  }) in
do vals <- return (10 :: 20 :: 30 :: []) in
do refs <- map (fun n -> !ref n) vals in
map (fun r -> !get r == 20) refs
\end{lstlisting}
Unbounded creation is the property that separates \cambria{} from the lexically scoped handler calculi of~\cite{biernacki2020binders}.
This program also uses parameters under let-polymorphism, applying the polymorphic \lstinline|map| to parameter-typed references.
\cambria{} is the first calculus in which this dynamic creation idiom enjoys a parametricity guarantee.
The body cannot distinguish handler implementations that differ in how \lstinline|$loc| is instantiated.

\subsection{Two Implementations of Probabilistic Urns} \label{sec:polya}

Parametrized handlers allow for different implementations of the same effect interface,
while preventing client code from inspecting the representation.
P\'olya's urn is a probabilistic process where an urn starts by containing a red ball and a blue ball.
After each sample from the urn, an additional ball of the sampled colour is placed back in the urn.
That is, in the following code, the value of \lstinline|b2| is more likely to match \lstinline|b1|.
\begin{lstlisting}
do urn <- !newurn () in
do b1 <- !sample urn in
do b2 <- !sample urn in
return (b1, b2)
\end{lstlisting}
In the probabilistic setting, \cambria{} has built-in effects \lstinline|!bernoulli : Double ~> Bool|, which returns \lstinline|true|
with probability given by the input clamped to $[0,1]$, and \lstinline|!uniform : Unit ~> Double|, which samples from the uniform distribution over $[0,1]$.
The former command can be used to handle the urn effect via the local state effect,
storing pairs of integers which track the number of red and blue balls.
\begin{lstlisting}
[$urn -> $loc] handler {
  newurn _ k -> k (!ref (1, 1)),
  sample urn k ->
    do (red, blue) <- !get urn in
    do isRed <- !bernoulli (red / (red + blue)) in
    do _ <- if isRed then !set (urn, (red + 1, blue))
                     else !set (urn, (red, blue + 1)) in
    k isRed
}
\end{lstlisting}
The parameter is instantiated to another parameter (\lstinline|$urn -> $loc|) rather than a concrete type,
composing the urn abstraction with the local state abstraction.

P\'olya showed that sampling from this urn admits a simple equivalent process~\cite{staton2018beta}:
draw a single bias $p$ uniformly from $[0,1]$ and then draw each sample independently from a Bernoulli distribution with that bias.
Hence, we can bypass local state to produce an implementation of P\'olya's urn by instantiating \lstinline|$urn| as a \lstinline|Double|.
This direct version is more efficient because it keeps no per-sample state and draws the bias once.
\begin{lstlisting}
[$urn -> Double] handler {
  newurn _ k -> k (!uniform ()),
  sample p k -> k (!bernoulli p)
}
\end{lstlisting}
Both handlers induce the same joint distribution on the returned pair.
The state-based handler gives $P(\text{both red}) = 1/2 \cdot 2/3 = 1/3$, while the direct handler gives $\mathbb{E}_{p\sim \Uniform[0,1]}[p^2] = 1/3$.
This is larger than the $1/4$ probability of independent, fair draws.
The client only sees \lstinline|$urn| and the operations \lstinline|!newurn| and \lstinline|!sample|.
By parametricity, it cannot inspect their implementation.
So, modulo a probabilistic lifting of the logical relation (see \Cref{sec:repind}), the implementations are indistinguishable by P\'olya's identity.

\subsection{Dynamically Threaded Concurrency} \label{sec:concurrency}

All of the previous examples treat each parameter as an independent effect instance.
Our final example of concurrent thread management~\cite{kammar2025equational} does not.
The \lstinline|!fork| operation uses multi-shot continuations to run both child and parent threads,
and the thread IDs pass between these processes while remaining abstract.
The parameter therefore coordinates distinct sub-computations rather than indexing an independent instance,
warranting a single handler to interpret all of them.
Another such effect is code jumps~\cite{staton2014jumps}, which appears as an example in the implementation.

A concurrent program has four operations:
\begin{itemize}
  \item \lstinline{!fork : Unit ~> $tid + Unit} duplicates the current process. The parent thread receives the thread ID of the child process. The child receives the unit value.
  \item \lstinline{!wait : $tid ~> Unit} stalls the current thread until the process of the provided thread ID finishes.
  \item \lstinline{!stop : Unit ~> Void} terminates the current thread.
  \item \lstinline{!act : Str ~> Unit}  performs a labelled observable action.
\end{itemize}
The thread IDs \lstinline|$tid| are abstract, so cannot be forged, compared, or inspected.
The following program constructs an N-shaped dependency between four actions, as visualized in \Cref{fig:deps}.

\begin{figure}
  \centering
  \includegraphics[width=5cm]{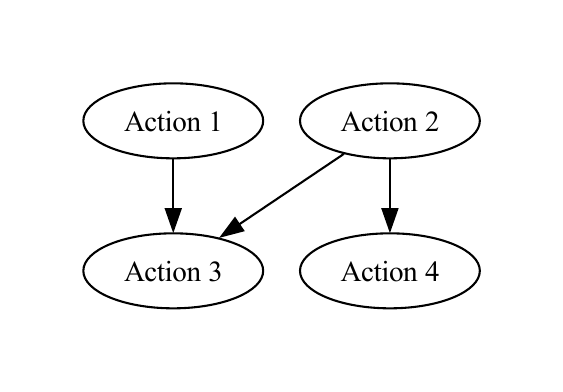}
  \caption{An N-shaped dependency between four actions, encoded by the dynamic thread management effect.}
  \label{fig:deps}
\end{figure}

\begin{lstlisting}
effect !fork : Unit ~> $tid + Unit.
effect !wait : $tid ~> Unit.
effect !stop : Unit ~> Void.
effect !act  : Str ~> Unit.
case !fork () of {
  inl a1 -> case !fork () of {
    inl a2 -> case !fork () of {
      inl a3 -> !wait a2 ; !act "Action 4" ; !stop (),
      inr _  -> !wait a1 ; !wait a2 ; !act "Action 3" ; !stop () },
    inr _  -> !act "Action 2" ; !stop () },
  inr _ -> !act "Action 1" ; !stop () }
\end{lstlisting}
We illustrate the expressiveness of parametrized handlers by handling this program in three different ways.
The first biases the child threads to run first, running the parent process once they are completed.
The second executes the parent process first, deferring each child process until it is waited on or the parent terminates.
The third handler is for diagnostics. It dispatches the concurrent operations to a secondary algebraic effect that encodes the action dependencies as a DAG.

In the child-biased implementation, the parent process only receives the child's thread ID after the child has terminated,
so any subsequent \lstinline|!wait| on that process is trivially satisfied.
Therefore, the \lstinline|!wait| operation becomes a no-op, and the thread IDs can be instantiated to \lstinline|Unit|.
The result of a returned thread is discarded, so the handler treats returning and stopping identically.
When handling the program above, the actions are executed in the order 1, 2, 3, 4.
\begin{lstlisting}
[$tid -> Unit] handler {
  return _ -> return (),
  act  s k -> !print s ; k (),
  wait _ k -> k (),
  stop _ k -> return (),
  fork _ k -> k (inr ()) ; k (inl ())
}
\end{lstlisting}
The parent-biased handler does not have the liberty of making \lstinline|!wait| a no-op.
It must run the child process if it hasn't already been executed.
An effective solution instantiates the thread ID as a thunk of the deferred child,
\lstinline|Unit -> Unit ! {...}|, where \lstinline|{...}| is shorthand for an effect signature left to be inferred.
The thunk is guarded by a local boolean reference ensuring the child runs at most once.
This instantiation may appear circular, because the captured continuation \lstinline|k|, which may invoke the concurrency operations, is called within the thunk.
But the continuation is resumed under the same handler, so those operations are caught, and the inferred effect signature only contains
\lstinline|!print| and the local-state effects \lstinline|!ref|, \lstinline|!set|, and \lstinline|!get|.
The latter can be handled similarly to those previously in \Cref{sec:local-state}, illustrating how parametrized handlers compose.
\begin{lstlisting}
[$tid -> Unit -> Unit ! {...}] handler {
  return _ -> return (),
  act  s k -> !print s ; k (),
  wait t k -> t () ; k (),
  stop _ k -> return (),
  fork _ k ->
    do hasRun <- !ref false in
    do thunk <- return (fun _ -> if !get hasRun then return ()
      else ( !set (hasRun, true) ; k (inr ()) )) in
    ( k (inl thunk) ; thunk () )
}
\end{lstlisting}
The \lstinline|wait| clause now runs its thunk before resuming, forcing the child if it hasn't already run.
At the end of the \lstinline|fork| clause, the child thunk is forced one last time to ensure it runs even if nothing has waited on it.
Returning and stopping are treated identically as before.

Applied to the same concurrent program above, the handler executes the actions in the order 2, 4, 1, 3.
This is different to the child-biased handler, but still consistent with the dependency diagram in \Cref{fig:deps}.
This second implementation relies on instantiating the thread ID type differently per handler.

If instead of running the program, we want a visual display of the action dependencies, we can implement a handler that converts the program into a DAG in Graphviz format.
We handle the concurrency effect into an intermediate pomset effect~\cite{kammar2025equational} that is parametrized by a type of nodes.
The only operation in this effect is \lstinline|!node : List $node * Str ~> $node|, which adds an extra node to the DAG.
The first input is the list of dependencies of the node and the second is the node's label.
The result type is the newly created node's identifier.
The handler is implemented by parameter passing the list of current dependencies through the computation.
Informally, each operation is implemented as follows:
\lstinline|act| calls the \lstinline|node| command with the current dependencies and provided label and passes on the newly created node,
\lstinline|wait| extends the current dependencies with the provided dependencies, \lstinline|return| and \lstinline|stop| forward the dependencies,
and \lstinline|fork| runs both the parent and child thread with the current dependencies, passing the thread ID of the child to the parent.
\begin{lstlisting}
[$tid -> List $node] handler {
  return _   -> return (fun ds -> return ds),
  act  s k   -> return (fun ds -> do a <- !node (ds, s) in k () (a :: [])),
  wait ds' k -> return (fun ds -> k () (concat ds ds')),
  stop _ k   -> return (fun ds -> return ds),
  fork _ k   -> return (fun ds -> do a <- k (inr ()) ds in k (inl a) ds),
  finally f  -> f []
}
\end{lstlisting}
The node effect can be handled by printing the DAG represented by the \lstinline|node| operations in Graphviz format.
Rather than using a user-defined implementation of the \lstinline|fresh| effect, we use \cambria{}'s built-in one,
which generates fresh instances of the \lstinline|Name| type.
The built-in primitive \lstinline|hash| converts each unique name to a string. The result is exactly the dependency visualization in \Cref{fig:deps}.
\begin{lstlisting}
[$node -> Name] handler {
  return _ -> return (fun _ -> return ()),
  node (ds, s) k -> return (fun _ -> do a <- !fresh () in (
    !print (hash a ++ " [label=\"" ++ s ++ "\"]") ; 
    map (fun d -> !print (hash d ++ " -> " ++ hash a)) ds ;
    k a () )),
  finally f -> !print "digraph {" ; f () ; !print "}"
}
\end{lstlisting}
Each of these three handlers instantiates the thread ID type to track completely different information.

\section{The Calculus}\label{sec:calculus}

The calculus we present for \cambria{} includes only the constructs strictly necessary to expose the metatheory of parametrized handlers
(parameters, parameter substitutions, and the parametrized handle form) along with general recursion and let-polymorphism.
The implementation of \cambria{} extends this calculus with standard data types and several surface conveniences, as discussed in \Cref{sec:surface}.
The examples of \Cref{sec:overview} make free use of these extensions.

\subsection{Types and Syntax} \label{sec:syntax}

The types of the language are inductively generated by the following grammar.
\[
\begin{array}{lrcl}
  \text{Value types}           & A, B   & ::= & \alpha \mid \psi \mid \tunit \mid \tfun{A}{\cC} \mid \thandler{\cC}{\cD} \\[4pt]
  \text{Computation types}     & \cC, \cD   & ::= & \comp{\forall \vec\alpha . A}{\Sig} \\[4pt]
  \text{Effect signatures}     & \Sig   & ::= & \emptyset \mid \Sig, \op : A \to B \\[4pt]
  \text{Type variable context} & \Theta & ::= & \eenv \mid \Theta, \alpha \\[4pt]
  \text{Parameter context}     & \Psi   & ::= & \eenv \mid \Psi, \psi \\[4pt]
  \text{Term context}          & \G     & ::= & \eenv \mid \G,\; x : \forall \vec{\alpha} . A
\end{array}
\]
There are three contexts: $\Theta$ for type variables, $\Psi$ for parameter variables, and $\Gamma$ for term variables.
The only extension to the type system beyond a standard handler calculus is the parameter variables $\psi$.
Parameter variables can be thought of as existential type variables, because they represent an implementation detail of the effect interface.
We do not allow universal quantification in operation arities, which prevents effects like heterogeneous state.
This scope permits the annotation completeness result for type inference (\Cref{sec:inference}), which would not hold in the polymorphic arity setting.
We foresee no issues extending the type-safety and parametricity results to polymorphic arities.

The syntax of values and computations is the following.
\[
\begin{array}{lrcl}
  \text{Values}       & v, w & ::= & x \mid \unit \mid \rec f\;x \mapsto c \mid h \\[4pt]
  \text{Handlers}      & h    & ::= &
    \hhandler\;\{
      \ret x \mapsto c_r,\;
      \op_1\; x\; k \mapsto c_1,\; \cdots,\;
      \op_m\; x\; k \mapsto c_m
  \} \\[4pt]
  \text{Computations} & c, d & ::= & \ret v \mid \op(v;\;x.c)\mid \ddo x \leftarrow c \iin d 
                                     \mid\; v\;w \mid \with \sigma, v \handle c \\[4pt]
  \text{Parameter subst.} & \sigma & ::= & \emptyset \mid \sigma[\psi \mapsto A]
\end{array}
\]
We write $\fun x \mapsto c$ as shorthand for $\rec f\;x \mapsto c$ when $f \notin \fv(c)$, where $\fv(c)$ are the free variables in $c$.
Similarly, we write $\hhandler\;\{ \ret x \mapsto c_r,\; \op\; x\; k \mapsto c_\op \}_{\op\in\Sig}$ as shorthand for the handler syntax where
$\Sig = \{\op_1 : A_1 \to B_1, \; \cdots\;, \op_m : A_m \to B_m \}$.
The key novelty is the \emph{parameter substitution} $\sigma$ in the handle term,
which bridges the abstract types in the body to their concrete representations in the handler.
We write $\dom(\sigma)$ for the parameter variables bound by $\sigma$, $\cod(\sigma)$ for the types they map to,
and $\fpv(X)$ for the free parameter variables in $X$.

\subsection{Surface Syntax}\label{sec:surface}

The \cambria{} surface language, used for the examples of \Cref{sec:overview}, departs from the above calculus in a few places.
These are primarily conveniences for the programmer, being more readable in practice.
Each maps directly onto a calculus term, so the metatheory is unaffected.

\paragraph{Finally clauses.}
The surface admits an optional $\mathsf{finally}\;x \mapsto c_f$ clause that runs after the handler returns.
A handler $h$ with a finally clause desugars at handle-time: $\with \sigma, h \handle c$ becomes $\ddo x \leftarrow (\with \sigma, \nofin{h} \handle c) \iin c_f$,
where $\nofin{h}$ is $h$ without the finally clause.
The result is a handle within a bind, so the metatheory carries through.

\paragraph{Operation syntax and effect declarations.}
The calculus writes an operation invocation as $\op(v;\;x.c)$, packaging the continuation $x.c$ syntactically.
The surface uses the generic effect~\cite{plotkin_algebraic_2003} $!\op\;v$, which desugars to $\op(v;\;x.\ret x)$. Hence, the calculus form $\op(v; \; x.c)$ is equivalent to
\[
  \ddo x \leftarrow {!\op}\;v \iin c.
\]
The surface additionally supports a local arity declaration \lstinline|effect !op : A ~> B. c| used only by type inference and discussed in \Cref{sec:inference}.

\paragraph{Surface conveniences.}
The remaining surface forms are sugar over the calculus.
These extensions are orthogonal to the metatheory, with parametricity (\Cref{thm:ft}), type safety (\Cref{cor:safety}),
and inference (\Cref{thm:relevant-complete}) transferring unchanged.
\begin{itemize}
  \item \emph{Pattern matching.} Irrefutable patterns $p ::= x \mid \_ \mid (p_1, p_2)$ are admitted in $\mathsf{do}$-bindings, function arguments, and operation clauses.
    Pair patterns desugar to projections.
  \item \emph{Multi-argument functions.} $\fun x_1\,x_2\,\ldots\,x_n \mapsto c$ and $\rec f\;x_1\,\ldots\,x_n \mapsto c$ desugar to nested unary abstractions.
  \item \emph{Sequencing.} $c_1\,;\,c_2$ desugars to $\ddo \_ \leftarrow c_1 \iin c_2$.
  \item \emph{Infix operators.} Arithmetic ($+$, $-$, $*$, $/$), equality ($==$), string concatenation ($+\!+$),
    and cons ($\,::\,$) desugar to applications of named primitives.
  \item \emph{Implicit $\mathsf{do}$.} A computation in a value position is hoisted into a fresh $\mathsf{do}$-binding around its enclosing computation.
    These are hoisted in left-to-right order. For example, $c_f \; c_x$ elaborates to $\ddo f \leftarrow c_f \iin (\ddo x \leftarrow c_x \iin f \; x)$.
  \item \emph{Default return clause.} A handler with no explicit $\mathsf{return}$ clause uses the identity $\ret x \mapsto \ret x$.
  \item \emph{Built-in types and operations.} The implementation provides $\tbool$, $\tint$, $\tdouble$, $\tstring$, $\tname$, $\void$,
    products, sums, lists, and maps, along with a small library of primitive operations including \lstinline|!fresh|, \lstinline|!print|, \lstinline|!bernoulli|, and \lstinline|!uniform|.
\end{itemize}

\subsection{Typing Rules}\label{sec:typing}

\newcommand{\calcsize}{\small\renewcommand{\TirNameStyle}[1]{\footnotesize\textsc{##1}}}

\begin{figure}
  \framebox{
    \begin{minipage}{0.95\linewidth}
      Type judgements.
      {\calcsize
      \begin{mathpar}
        \inferrule{\alpha \in \Theta}{\tj \Theta \Psi \alpha}

        \inferrule{\psi \in \Psi}{\tj \Theta \Psi \psi}

        \inferrule{ }{\tj \Theta \Psi \tunit}

        \inferrule{ \tj \Theta \Psi A \and \tj \Theta \Psi \cC }{\tj \Theta \Psi \tfun{A}{\cC}}

        \inferrule{ \tj \Theta \Psi \cC \and \tj \Theta \Psi \cD }{\tj \Theta \Psi \thandler{\cC}{\cD}}

        \inferrule{ [\tj \Theta \Psi {A_\op} \quad \tj \Theta \Psi {B_\op} ]_{(\op : A_\op \to B_\op) \in \Sig} }{\tj \Theta \Psi \Sig}

        \inferrule{ \tj \Theta \Psi A \and \tj \Theta \Psi \Sig }{\tj \Theta \Psi {\comp{A}{\Sig}}}

        \inferrule{ \tj {\Theta, \vec{\alpha}} \Psi A }{\tj \Theta \Psi {\forall\vec\alpha.A}}

        \inferrule{ [\tj \Theta \Psi {\forall\vec\alpha.A}]_{(x:\forall \vec\alpha.A) \in \Gamma} }{\tj \Theta \Psi \Gamma}
      \end{mathpar}
      }
      Value judgements.
      {\calcsize
      \begin{mathpar}
        \inferrule*[right=Var]
        {(x : \forall \vec{\alpha}.B) \in \G \\ [\tj \Theta \Psi {A_\alpha}]_{\alpha\in\vec\alpha}}
        {\vj{\Theta}{\Psi}{\G}{x}{B[A_\alpha / \alpha]_{\alpha\in\vec\alpha}}}

        \inferrule*[right=Unit]
          { }
          {\vj{\Theta}{\Psi}{\G}{\unit}{\tunit}}

        \inferrule*[right=Rec]
          {\cj{\Theta}{\Psi}{\G, f : \tfun{A}{\comp{B}{\Sig}},\, x : A}{c}{B}{\Sig}}
          {\vj{\Theta}{\Psi}{\G}{\rec f\; x \mapsto c}{\tfun{A}{\comp{B}{\Sig}}}}

        \inferrule*[right=Handler]{
          \cj{\Theta}{\Psi}{\G, x : A}{c_r}{B}{\Sig'}
          \\
          \Sig \setminus \Sig'' \subseteq \Sig'
          \\\\
          \left[
            \cj{\Theta}{\Psi}{\G,\, x : A_\op,\, k : \tfun{B_\op}{\comp{B}{\Sig'}}}{c_\op}{B}{\Sig'}
          \right]_{(\op:A_\op \to B_\op) \in \Sig''}
          }
          {\vj{\Theta}{\Psi}{\G}{
            \hhandler\;\{\ret x \mapsto c_r,\;
              \op\; x\; k \mapsto c_\op
            \}_{\op \in\Sig''}
          }{\thandler{\comp{A}{\Sig}}{\comp{B}{\Sig'}}}}
      \end{mathpar}
      }
      Computation judgements.
      {\calcsize
      \begin{mathpar}
        \inferrule*[right=Return]
        {\vj{\Theta}{\Psi}{\G}{v}{A}}
        {\cj{\Theta}{\Psi}{\G}{\ret v}{A}{\Sig}}

        \inferrule*[right=Gen]{
          \cj {\Theta, \vec{\alpha}} \Psi \G c A \Sig
        }{
          \cj \Theta \Psi \G c {\forall \vec{\alpha}.A} \Sig
        }

        \inferrule*[right=App]
        {\vj{\Theta}{\Psi}{\G}{v_1}{\tfun{A}{\cC}}
          \\
          \vj{\Theta}{\Psi}{\G}{v_2}{A}
          }
          {\cjp{\Theta}{\Psi}{\G}{v_1\; v_2}{\cC}}

        \inferrule*[right=Do]
        {\cj{\Theta}{\Psi}{\G}{c_1}{\forall \vec{\alpha}. A}{\Sig}
          \\
          \cj{\Theta}{\Psi}{\G, x : \forall\vec{\alpha}.A}{c_2}{B}{\Sig}
          }
          {\cj{\Theta}{\Psi}{\G}{\ddo x \leftarrow c_1 \iin c_2}{B}{\Sig}}

        \inferrule*[right=Op]
          {(\op : \arity{A_{\op}}{B_{\op}}) \in \Sig
          \\
          \vj{\Theta}{\Psi}{\G}{v}{A_{\op}}
          \\
          \cj{\Theta}{\Psi}{\G,\, x : B_{\op}}{c}{A}{\Sig}
          }
          {\cj{\Theta}{\Psi}{\G}{\op(v;\;x.c)}{A}{\Sig}}

        \inferrule*[right=ParHandle]{
          \vj{\Theta}{\Psi}{\G}{v}{\thandler{\cC[\sigma]}{\cD}}
          \\
          \cjp{\Theta}{\Psi, \dom(\sigma)}{\G}{c}{\cC}
          \\
          \left[\tj \Theta \Psi A\right]_{A\in\cod(\sigma)}
          }{
            \cjp{\Theta}{\Psi}{\G}{\with \sigma, v \handle c}{\cD}
          }
      \end{mathpar}
      }
    \end{minipage}
  }
  \caption{Typing rules for \cambria{}.}
  \label{fig:typing-rules}
\end{figure}

The typing rules for \cambria{} are given in \Cref{fig:typing-rules}.
A type judgement $\tj \Theta \Psi A$ states that the type $A$ is well-formed in type variable context $\Theta$ and parameter context $\Psi$.
Value judgements are of the form $\vj \Theta \Psi \Gamma v A$ and computation judgements are of the form $\cj \Theta \Psi \Gamma c A \Sig$.
Both implicitly enforce that $\tj \Theta \Psi \Gamma$ and $\tj \Theta \Psi A$, and the computation judgement additionally enforces that $\tj \Theta \Psi \Sig$.
In particular, no $\alpha \in \vec\alpha$ appears free in $\Gamma$ or $\Sig$ in the \textsc{Gen} rule.
This permits let-polymorphism without any value restriction~\cite{kammar_no_2017}.
Moreover, parameters are never generalized and remain rigid within the body of a handler so that the handler can instantiate them consistently.
Otherwise the abstraction boundary would break.
In the \textsc{Handler} rule, handled operations do not need to cover the full set of operations in the computation effect signature.
The remaining operations are \emph{forwarded}, which is indicated by $ \Sig \setminus \Sig'' \subseteq \Sig' $ in the premise.
The continuation captured by an operation is bound as an ordinary variable,
so an operation clause may resume it any number of times, permitting \emph{multi-shot continuations}.

The only non-standard type judgement is the \textsc{ParHandle} rule.
The rule mirrors a classic handle rule, but the handler's input type matches the body's type only after the parameter substitution is applied to it.
When the parameter substitution is empty ($\sigma = \emptyset$), the rule becomes classic handler application.
When non-empty, the parameters introduced by $\sigma$ remain confined to the handled body.
That is, for $\psi \in \dom(\sigma)$, $\psi \notin \fpv(\cD, \Gamma)$, since $\cD$ and $\Gamma$ are implicitly well-formed in $\Theta;\Psi$.
An operation whose arity mentions $\psi$ may still be forwarded, but it is then exposed at its instantiated arity, so $\psi$ does not escape.
Our examples always handle such operations, but the metatheory does not require it.

\subsection{Operational Semantics and Type Safety}

Evaluation contexts identify where the next reduction step occurs.
\[
  E \;::=\; \square
    \mid \ddo x \leftarrow E \iin c
    \mid \with \sigma, v \handle E
\]
Single-step reduction $c \leadsto c'$ is defined below, closed under the congruence rule: if $c \leadsto c'$ then $E[c] \leadsto E[c']$.
\begin{align*}
  (\rec f\;x \mapsto c)\;v
    &\quad\leadsto\quad c[v/x,\;(\rec f\;x \mapsto c)/f] \\[-3pt]
  \ddo x \leftarrow \ret v \iin c
    &\quad\leadsto\quad c[v/x] \\[-3pt]
  \ddo x \leftarrow \op(v; y.c) \iin c'
    &\quad\leadsto\quad \op(v; y.\ddo x \leftarrow c \iin c')
\end{align*}
For handler reduction, let $h = \hhandler\;\{\ret x \mapsto c_r,\;\op\;x\;k \mapsto c_\op\}_{\op\in\Sig}$. Then handle terms reduce as follows.
\begin{align*}
  \with \sigma, h \handle \ret v
    &\quad\leadsto\quad c_r[v/x] \\[-3pt]
  \with \sigma, h \handle \op(v;\;y.c)
    &\quad\leadsto\quad c_\op[v/x,\;(\fun y \mapsto \with \emptyset, h \handle c)/k]
    \tag{$\op \in \Sig$} \\[-3pt]
  \with \sigma, h \handle \op(v;\;y.c)
    &\quad\leadsto\quad \op(v;\;y.\;\with \emptyset, h \handle c)
    \tag{$\op \notin \Sig$}
\end{align*}
Reduction is deterministic because reducible computations decompose uniquely into an evaluation context and a redex.
The continuation supplied to the operation clause reinstates the handler, giving \emph{deep handling}~\cite{kammar_handlers_2013}.
In reinstated continuations, we discard the parameter substitution because it plays no role in reduction.
This is the formal sense in which parameter variables are erased at runtime. There is no dynamic dispatch, no coercions, and no type-directed reduction.
The parameter annotations are only used for type checking.

\begin{proposition}[Parameter Erasure] \label{prop:erasure}
  Define a parameter erasure map $\lfloor {-} \rfloor$ as
  \[
    \lfloor \with \sigma, h \handle c \rfloor = \with \emptyset, \lfloor h \rfloor \handle \lfloor c \rfloor
  \]
  and structural identity otherwise. Then $c$ is reducible if and only if $\lfloor c \rfloor$ is reducible.
  Moreover, if $c \leadsto c'$ then $\lfloor c \rfloor \leadsto \lfloor c' \rfloor$.
\end{proposition}

\begin{proof}
  No rule inspects $\sigma$, so erasing the parameter substitutions gives a valid reduction by the same rule, establishing both results.
\end{proof}

Type safety of a language states that well-typed terms do not get stuck.
In the context of effect systems, this means a well-typed computation has reduced to a value,
is \emph{control-stuck} on an unhandled operation in its effect signature, or can take a reduction step.
This result rests on a preservation lemma that usually says reduction preserves a program's type.
In \cambria{}, to keep reduction type-agnostic, \textsc{ParHandle} reductions discard the parameter substitution.
This gives a weaker notion of preservation, but still enforces type safety.

\begin{theorem}[Progress and Preservation] \label{thm:pp}
  If $\cj{\eenv}{\eenv}{\eenv}{c}{A}{\Sig}$, then either:
  \begin{enumerate}
    \item $c = \ret v$ for some $\vj{\eenv}{\eenv}{\eenv}{v}{A}$
    \item $c = \op(v;\; y.c')$ for some $\op : A_\op \to B_\op \in \Sig$, $\vj{\eenv}{\eenv}{\eenv}{v}{A_\op}$,
      and $\cj{\eenv}{\eenv}{y:B_\op}{c'}{A}{\Sig}$.
    \item $c \leadsto d$ and there is some $\cj{\eenv}{\eenv}{\eenv}{d'}{A}{\Sig}$ such that $\lfloor d' \rfloor = \lfloor d \rfloor$.
  \end{enumerate}
\end{theorem}

\begin{proof}
  We first show progress by induction over the computation derivations, with canonical forms forcing the reduction steps if not a value or control-stuck operation.
  For preservation, we prove type substitution, parameter substitution, and variable substitution lemmas by structural induction:
  \begin{itemize}
    \item If $\cj{\Theta,\alpha}{\Psi}{\Gamma}{c}{A}{\Sig}$ and $\Theta;\Psi \vdash B$, then $\cjp{\Theta}{\Psi}{\Gamma[B / \alpha]}{c[B/ \alpha]}{(\comp{A}{\Sig})[B/ \alpha]}$.
    \item If $\cj{\Theta}{\Psi,\psi}{\Gamma}{c}{A}{\Sig}$ and $\Theta;\Psi \vdash B$, then $\cjp{\Theta}{\Psi}{\Gamma[B / \psi]}{c[B/ \psi]}{(\comp{A}{\Sig})[B/ \psi]}$.
    \item If $\cj{\Theta}{\Psi}{\Gamma, x:\forall \vec{\alpha}.B}{c}{A}{\Sig}$ and $\vj{\Theta,\vec{\alpha}}{\Psi}{\Gamma}{v}{B}$, then $\cjp{\Theta}{\Psi}{\Gamma}{c[v / x]}{\comp{A}{\Sig}}$.
  \end{itemize}
  Then, for non-handler reduction, preservation follows from $d'$ instantiated to $d$.
  For handler reductions, the witness $d'$ reinstates the parameter substitution the reduction discarded.
  \begin{itemize}
    \item For $d = c_r[v / x]$, we take $d' = c_r[v[\sigma] / x]$.
    \item For $d = c_\op[v / x, (\fun y \mapsto \with \emptyset, h \handle c) / k]$, we take $d' = c_\op[v[\sigma] / x, (\fun y \mapsto \with \emptyset, h \handle c[\sigma]) / k]$.
    \item For $d = \op(v; \; y. \with \emptyset, h \handle c)$, we take $d' = \op(v[\sigma]; \; y. \with \emptyset, h \handle c[\sigma])$.
  \end{itemize}
  Then, since $\lfloor v[\sigma] \rfloor = \lfloor v \rfloor$ and $\lfloor c[\sigma] \rfloor = \lfloor c \rfloor$,
  and erasure commutes with substitution, $\lfloor d' \rfloor = \lfloor d \rfloor$ in all cases.
  By the substitution lemmas and \textsc{ParHandle} rule, these witnesses are well-typed at $\cj{\eenv}{\eenv}{\eenv}{d'}{A}{\Sig}$.
\end{proof}

In particular, if $\cj{\eenv}{\eenv}{\eenv}{c}{A}{\emptyset}$, then case (2) is vacuous and the result reduces to the computation being a value or taking a step.

\begin{corollary}[Type Safety]\label{cor:safety}
  If $\cj{\eenv}{\eenv}{\eenv}{c}{A}{\Sig}$ and $c \leadsto^* d$, then $d$ is a value, control-stuck on an operation in $\Sig$, or reducible.
\end{corollary}

\begin{proof}
  By induction on the reduction length, maintaining the invariant that every reduct carries a well-typed witness at $\comp{A}{\Sig}$ erasing to it.
  The base case is $c$, and \Cref{thm:pp} supplies the witness for each successor.
  \Cref{prop:erasure} and determinism force the actual reduct and the witness to share a parameter-erased successor,
  and \Cref{thm:pp} applied to the witness of $d$ provides the trichotomy.
\end{proof}

\section{Parametricity}\label{sec:parametricity}

This section formalizes the abstraction guarantee that parametrized handlers provide to client code.
We prove a parametricity result using a step-indexed binary logical relation (\Cref{thm:ft}).
Parametrized handlers decompose into classic handlers and existential types.
Thus, the logical relation reduces to a standard combination of handlers, universal types, existential types, and recursion, with each type constructor contributing an independent clause.
Since parameters are erased and their abstraction is captured by a single relation, the construction is direct,
needing neither a biorthogonal definition of the computation relation~\cite{biernacki2018handle},
nor Kripke worlds which are required for a full mutable heap~\cite{ahmed2009state}.

We formalize the decomposition as an elaboration into a core calculus (\Cref{sec:elaboration}) and prove the parametricity theorem on the core.
Adequacy transports the result back to \cambria{}.

\subsection{The Core Language}

Parametrized handlers can intuitively be thought of as first packing the concrete parameter instantiation into an existential type,
and then handling the abstract computation.
That is, a handle with a single parameter substitution elaborates along the following lines.
\begin{align*}
  \with \{\psi \mapsto A\}, v \handle c \quad \mapsto &\quad \unpack \pack(\psi, A, v) \as (\psi, x) \iin (\with x \handle c)
\end{align*}
This decomposition is formalized in \Cref{sec:elaboration}. Before proceeding, we introduce the syntax of the core language, the target language of the elaboration.
We update the handle syntax and extend the value types, values and computations with the following constructs in place of
the parametrized handle term ($\with \sigma, v \handle c$) in \Cref{sec:syntax}.
\[
\begin{array}{lrcl}
  \text{Value types}  & A, B & ::= & \cdots \mid \exists\psi.A \\[4pt]
  \text{Values}       & v, w & ::= & \cdots \mid \pack(\psi, A, v) \\[4pt]
  \text{Computations} & c, d & ::= & \cdots \mid \with v \handle c \mid \unpack v \as (\psi, x) \iin c
\end{array}
\]
Parameter substitutions $\sigma$ are no longer part of the syntax in the core language.
Core judgements, denoted with $\vdashc$, have a new typing derivation for existential types and, in place of the \textsc{ParHandle} rule,
we have a plain handle rule, as well as existential introduction and elimination rules.
\begin{mathpar}
  \inferrule*{\tjc \Theta {\Psi, \psi} A }{ \tjc \Theta \Psi {\exists\psi.A}}

  \inferrule*[right=Unpack]{
    \vjc{\Theta}{\Psi}{\G}{v}{\exists \psi.A}
    \\
    \vjc{\Theta}{\Psi,\psi}{\G, x:A}{c}{\cC}
  }{
    \cjpc{\Theta}{\Psi}{\G}{\unpack v \as (\psi, x) \iin c}{\cC}
  }

  \inferrule*[right=Handle]{
    \vjc{\Theta}{\Psi}{\G}{v}{\thandler{\cC}{\cD}}
    \\
    \cjpc{\Theta}{\Psi}{\G}{c}{\cC}
  }{
    \cjpc{\Theta}{\Psi}{\G}{\with v \handle c}{\cD}
  }

  \inferrule*[right=Pack]{
    \tjc{\Theta}{\Psi}{A}
    \\
    \vjc{\Theta}{\Psi}{\G}{v}{B[A / \psi]}
  }{
    \vjc{\Theta}{\Psi}{\G}{\pack(\psi, A, v)}{\exists \psi.B}
  }
\end{mathpar}
In the \textsc{Unpack} rule, $\psi \notin \fpv(\cC, \G)$ because $\tjc{\Theta}{\Psi}{\G}$ and $\tjc{\Theta}{\Psi}{\cC}$ implicitly hold in the conclusion.

The operational semantics replaces parametrized handler reduction with classic handler reduction and adds a new reduction rule for existential terms,
inheriting the remaining rules.
Let $h = \hhandler\;\{\ret x \mapsto c_r,\;\op\;x\;k \mapsto c_\op\}_{\op\in\Sig}$, then the core language reductions on the new constructs are as follows.
\begin{align*}
  \with h \handle \ret v
    &\quad\corestep_p\quad c_r[v/x] \\[-3pt]
  \with h \handle \op(v;\;y.c)
    &\quad\corestep_p\quad c_\op[v/x,\;(\fun y \mapsto \with h \handle c)/k]
    \tag{$\op\in\Sig$} \\[-3pt]
  \with h \handle \op(v;\;y.c)
    &\quad\corestep_p\quad \op(v;\;y.\;\with h \handle c)
    \tag{$\op\notin\Sig$} \\[-3pt]
  \unpack \pack(\psi, A, v) \as (\psi, x) \iin c
    &\quad\corestep_a\quad c[v / x]
\end{align*}
A classic handle evaluation frame replaces the parametrized one.
\[
  E \;::=\; \cdots \mid \with v \handle E
\]
Core language reduction is written $\corestep$, distinguishing it from \cambria{} reduction $\leadsto$.
We further annotate core reductions with $\corestep_p$ for \emph{principal} steps and $\corestep_a$ for \emph{administrative} steps.
The only administrative step is the existential elimination, which has no runtime content.
In \cambria{}, this distinction collapses into the single parametrized handle reduction, as formalized in \Cref{prop:op-corresp}.

\subsection{Step-Indexed Logical Relation and Fundamental Theorem}

A logical relation builds a model of a calculus to capture an explicit property.
Each type is interpreted as a relation between pairs of closed terms, capturing when two terms are indistinguishable at that type.
We build a step-indexed logical relation~\cite{appel_indexed_2001, ahmed_semantics_2006}
over the core language to show that contextual equivalence is parametric over the type variables and parameters.
The formal result is the fundamental theorem (\Cref{thm:ft}), which states that every well-typed term is related to itself
uniformly across all admissible interpretations of its free variables.
Thus, the parameters cannot be inspected or used outside of the interfaces to the effects.

As parametrized handlers decompose into well-known constructs, the clauses of the step-indexed logical relation are standard.
While we only introduce the relation on the core syntax,
it extends to empty, product, sum, and list types, and all of the features in the current version of \cambria{}.
This includes row polymorphism~\cite{leijen_koka_2014, lindley_row-based_2012}, which is already utilized for effect inference in the current type system.
Instantiating row variables to closed signatures gives the present relation, and row generalization is handled exactly like type generalization.
Since parameters and row variables are disjoint with independent substitution, the two features compose.
For features not yet supported in the implementation, such as recursive types~\cite{ahmed_semantics_2006}
and universal quantification in operation arities~\cite{sekiyama2019handling},
we expect the parametricity result to generalize given the separation of concerns between parameters, type variables, and step-indexing.
We leave this to future work.

The step-indexing is necessary for general recursion, and integrates well with handlers.
Each operation observation strictly decreases the index, making the self-reference well-founded without resorting to induction over the size of the computation trees.
Since captured continuations are bound as ordinary variables and the relation is downward-closed in the step index,
each resumption is related, and multi-shot continuations need no special treatment.
The step index tracks whether two terms can be discerned within that number of computation steps, so the relation is refined as the step index grows.

\begin{definition}
  An \emph{admissible relation} is a step-indexed relation on closed value terms
  $R \subseteq \mathsf{Val}_0 \times \mathsf{Val}_0 \times \mathbb{N}$, an element of which we write $(v_1, v_2)_n$,
  that is downward-closed: if $(v_1, v_2)_n \in R$ and $n' \leq n$, then $(v_1, v_2)_{n'} \in R$.
  We write $\Rel$ for the set of all admissible relations.
\end{definition}

Given a type context $\Theta$ and parameter context $\Psi$,
a \emph{relational interpretation} $\rho : \Theta \cup \Psi \to \Rel$ assigns an admissible relation to each type variable and parameter.
The logical relation assigns to each typing judgement $\tjc \Theta \Psi X$ a mapping from relational interpretations to an admissible relation, written as $\interp{X}$.
When the type variable and parameter contexts are closed ($\Theta \cup \Psi = \emptyset$),
the interpretation is the unique map on the empty domain, so we drop the subscript.
The logical relation is defined inductively over the typing judgements and step index.
\begin{align*}
  \interp{\alpha} &= \rho(\alpha) \\
  \interp{\psi} &= \rho(\psi) \\
  \interp{\tunit} &= \{(\unit, \unit)_n \mid n \in \mathbb{N}\} \\
  \interp{\tfun{A}{\cC}} &= \{(v_1, v_2)_n \mid \forall i \leq n.\; \forall (w_1, w_2)_i \in \interp{A}.\; (v_1\;w_1, v_2\;w_2)_i \in \interp{\cC}\} \\
  \interp{\thandler{\cC}{\cD}} &= \{(v_1, v_2)_n \mid \forall i \leq n.\; \forall (c_1, c_2)_i \in \interp{\cC}. \\
  &\qquad\qquad\qquad (\with v_1 \handle c_1, \with v_2 \handle c_2)_i \in \interp{\cD}\} \\
  \interp{\exists\psi.A} &= \{(\pack(\psi, B_1, v_1), \pack(\psi, B_2, v_2))_n \mid \exists R \in \Rel.\; (v_1, v_2)_n \in \interp[{\rho[\psi \mapsto R]}]{A}\} \\
  \interp{\forall\vec\alpha.A} &= \{(v_1, v_2)_n \mid \forall \rho' \in \Rel^{\vec\alpha}.\; (v_1, v_2)_n \in \interp[{\rho[\rho']}]{A}\} \\
  \interp{\G} &= \{(\gamma_1, \gamma_2)_n \mid \forall(x:\forall\vec\alpha.A)\in\Gamma. \; (\gamma_1(x), \gamma_2(x))_n \in\interp{\forall\vec\alpha.A} \}
\end{align*}
Two closed computations are related at step index $n$ if they remain observationally equivalent after $n$ reduction steps.
That is, $(c_1,c_2)_n \in \interp{\comp{A}{\Sig}}$ if for all $i < n$:
\begin{itemize}
  \item If $c_1 \corestep^i \ret v_1$, then there is $v_2$ such that $c_2 \corestep^* \ret v_2$ and $(v_1, v_2)_{n-i} \in \interp{A}$
  \item If $c_1 \corestep^i \op(v_1; y.d_1)$ for $(\op : \arity{A_\op}{B_\op}) \in \Sig$,
    then there are $v_2$ and $d_2$ such that $c_2 \corestep^* \op(v_2; y.d_2)$ and:
  \begin{itemize}
    \item $(v_1, v_2)_{n-i} \in \interp{A_\op}$
    \item $\forall j < n{-}i.\; \forall (w_1, w_2)_j \in \interp{B_\op}.\; (d_1[w_1/y], d_2[w_2/y])_j \in \interp{\comp{A}{\Sig}}$
  \end{itemize}
\end{itemize}
The step index tracks reductions on the left, giving an asymmetric approximating relation.
The self-reference of $\interp{\comp{A}{\Sig}}$ is well-founded because the index strictly reduces.
As in the value case, computations are related at a universally quantified return type if they are related for every relational interpretation of the bound variables.
\[
  \interp{\comp{\forall\vec\alpha.A}{\Sig}}
    = \{(c_1, c_2)_n \mid \forall \rho' \in \Rel^{\vec\alpha}.\; (c_1, c_2)_n \in \interp[{\rho[\rho']}]{\comp{A}{\Sig}}\}
\]
The fundamental theorem states that this relation holds for every well-typed term.

\begin{theorem}[Fundamental Theorem]\label{thm:ft}
  If $\cjpc{\Theta}{\Psi}{\G}{t}{X}$, then for all $\rho : \Theta \cup \Psi \to \Rel$, $n\in\mathbb{N}$,
  and $(\gamma_1, \gamma_2)_n \in \interp{\G}$,
  \[
    (t[\gamma_1], t[\gamma_2])_n \in \interp{X}.
  \]
\end{theorem}

\begin{proof}
  The proof rests on three lemmas. For each, let $\rho : \Theta \cup \Psi \to \Rel$ be a relational interpretation for type contexts $\Theta; \Psi$.
  \begin{itemize}
    \item \textbf{Monotonicity:} For all $\tjc \Theta \Psi X$, if $(x_1, x_2)_n \in \interp X$ and $n' \leq n$, then $(x_1, x_2)_{n'} \in \interp X$.
      This is proved by induction over the definition of the logical relation $\interp{{-}}$.
    \item \textbf{Anti-reduction:} For all $\tjc \Theta \Psi \cC$, if $c_1 \corestep^i c_1'$, $c_2 \corestep^* c_2'$,
      and $(c_1', c_2')_n \in \interp \cC$, then $(c_1, c_2)_{n+i} \in \interp \cC$. This holds because reduction is deterministic.
    \item \textbf{Substitution:} For $\tjc {\Theta, \alpha} \Psi X$ and $\tjc \Theta \Psi A$, $\interp {X[A / \alpha]} = \interp[{\rho[\alpha \mapsto \interp{A}]}]{X}$.
      Equivalently for substitution of parameter variables. This is also proved by induction over the definition of the logical relation.
  \end{itemize}
  Then the fundamental theorem is proved by induction over the value and computation typing judgements.
  The insight that parametrized handlers decompose into classic handlers and existential types makes all of these cases standard.
  Nevertheless, we highlight the key ones.

  \textbf{Rec} and \textbf{Handler:} By inner induction on $n$ with the same pattern.
  The inner IH gives the recursive occurrences related at $n{-}1$, the structural IH relates the bodies after a step, and anti-reduction lifts this back to $n$.
  For the handler, the return case is straightforward.
  For an operation $\op:\arity{A_\op}{B_\op}\in\Sig$, the handle term reduces either to its operation clause applied to the deep continuation
  ($\fun y\mapsto\with h \handle c$) if handled, or to $\op(v;\;y.\with h \handle c)$ if forwarded.
  In both cases the inner IH relates the continuations one step index lower.

  \textbf{Do:} Via the bind lemma (\Cref{lem:bind}).

  \textbf{Gen:} For any $\rho' \in \Rel^{\vec\alpha}$, using $\vec\alpha \notin \fv(\Gamma, \Sig)$, we can weaken to
  $(\gamma_1,\gamma_2)_n \in \interp[{\rho[\rho']}]{\Gamma}$. Then we use the IH to conclude the computations are in $\interp[{\rho[\rho']}]{\comp{A}{\Sig}}$.

  \textbf{Pack:} The IH gives $(v[\gamma_1], v[\gamma_2])_n \in \interp{B[A/\psi]}$.
  Taking $R = \interp{A}$ as the existential witness, admissible by monotonicity, the packed values are related at $\interp{\exists\psi.B}$ by the substitution lemma.

  \textbf{Unpack:} The IH on $v$ yields packed witnesses $\pack(\psi, B_k, w_k)$ (for $k\in\{1,2\}$), together with some
  $R \in \Rel$ relating $w_1, w_2$ at $\interp[{\rho[\psi\mapsto R]}]{A}$.
  The IH on $c$ at the extended contexts puts the bodies in $\interp[{\rho[\psi\mapsto R]}]{\cC} = \interp{\cC}$ (using $\psi \notin \fpv(\cC)$),
  and anti-reduction lifts this to the unpack terms.
\end{proof}

\begin{lemma}[Bind Lemma] \label{lem:bind}
  If $(c_1, c_2)_n \in \interp{\comp{\forall \vec{\alpha}.A}{\Sig}}$ and for all $i \leq n$ and $(v_1, v_2)_i \in \interp{\forall \vec{\alpha}.A}$
  we have $(d_1[v_1/x], d_2[v_2/x])_i \in \interp{\comp{B}{\Sig}}$, then
  \[
    (\ddo x \leftarrow c_1 \iin d_1, \ddo x \leftarrow c_2 \iin d_2)_n \in \interp{\comp{B}{\Sig}}.
  \]
\end{lemma}

\begin{proof}
  We prove the lemma via strong induction on $n$.
  If $c_1 \corestep^i \ret v_1$ for some $i < n$, then the bind reduces to $d_1[v_1/x]$ after $i{+}1$ steps.
  Then $c_2 \corestep^* \ret v_2$ with $(v_1, v_2)_{n-i} \in \interp{\forall\vec\alpha.A}$,
  the lemma premise gives $(d_1[v_1/x], d_2[v_2/x])_{n-i-1} \in \interp{\comp{B}{\Sig}}$, and anti-reduction and monotonicity lift this back to $n$.
  If instead $c_1 \corestep^i \op(v_1;\; y.c_1')$ for some $i < n$, the bind reduces to $\op(v_1;\; y. \ddo x \leftarrow c_1' \iin d_1)$ after $i{+}1$ steps.
  Then $c_2 \corestep^* \op(v_2;\; y.c_2')$ with $(v_1, v_2)_{n-i} \in \interp{A_\op}$ and continuations related for $j < n{-}i$.
  The inner IH applies to $(\ddo x \leftarrow c_1'[w_1/y] \iin d_1, \ddo x \leftarrow c_2'[w_2/y] \iin d_2)$ at this strictly lower index,
  and anti-reduction again lifts this back to $n$.
\end{proof}

\subsection{Elaboration from \cambria{} and Operational Correspondence} \label{sec:elaboration}

\newcommand{\elab}[1]{{(#1)^\dagger}}

The fundamental theorem applies to the core language.
To transfer parametricity to \cambria{}, we need to connect the two calculi via an elaboration.
We define such an elaboration and show that it is faithful.
Hence, the fundamental theorem is reflected by \cambria{}, and the abstract interfaces of the parametrized handlers are formally guaranteed.

\begin{definition}
  The \emph{elaboration translation} $\elab{{-}}$ from \cambria{}'s syntax to the core syntax is given by
  structural identity on all terms, apart from the parametrized handle clause, where it is defined inductively over the size of the parameter substitution.
  \begin{align*}
    \elab{\with \emptyset, v \handle c}              &= \with v^\dagger \handle c^\dagger \\[-2pt]
    \elab{\with \sigma[\psi \mapsto A], v \handle c} &= \unpack \pack(\psi, A, v^\dagger) \as (\psi, x) \iin \elab{\with \sigma, x \handle c}
  \end{align*}
\end{definition}

This elaboration preserves typing judgements, and is thus well-defined.

\begin{proposition}[Type Preservation] \label{prop:preservation}
  If $\cjp \Theta \Psi \Gamma c \cC$, then $\cjpc \Theta \Psi \Gamma {c^\dagger} \cC$.
\end{proposition}

\begin{proof}
  The only non-trivial case is the \textsc{ParHandle} elaboration. This is proved by inner induction on the parameter substitution.
  When $\sigma=\emptyset$, the rule is exactly the classic handler rule.
  When $\sigma=\sigma'[\psi \mapsto A]$, unwinding the typing derivation gives:
  \[
    \inferrule{
      \vj{\Theta}{\Psi}{\G}{v}{\thandler{\cC[\sigma'[\psi\mapsto A]]}{\cD}}
      \\
      \cjp{\Theta}{\Psi,\dom(\sigma'), \psi}{\G}{c}{\cC}
      \\
      \tj \Theta \Psi A \quad \cdots
    }{
      \cjp{\Theta}{\Psi}{\G}{\with \sigma'[\psi \mapsto A], v \handle c}{\cD}
    }
  \]
  After invoking the structural IH, as $\psi \notin \fpv(\cD)$, we can derive
  \begin{equation}
    \inferrule{
      \vjc{\Theta}{\Psi}{\G}{v}{(\thandler{\cC[\sigma']}{\cD})[A / \psi]}
    }{
      \vjc{\Theta}{\Psi}{\G}{\pack(\psi, A, v)}{\exists\psi.(\thandler{\cC[\sigma']}{\cD})}
    }\text{,} \label{eq:pack}
  \end{equation}
  and weakening $\cjpc{\Theta}{\Psi,\dom(\sigma'), \psi}{\G}{c}{\cC}$ gives,
  \begin{equation}
    \cjpc \Theta {\Psi, \psi} {\Gamma, x : \thandler{\cC[\sigma']}{\cD}} {\with \sigma', x \handle c} \cD . \label{eq:freepsi}
  \end{equation}
  Then, combining \eqref{eq:pack} and \eqref{eq:freepsi} we can derive
  \[
    \cjpc{\Theta}{\Psi}{\G}{\unpack \pack(\psi, A, v) \as (\psi, x) \iin \elab{\with \sigma', x \handle c}}{\cD}
  \]
  where we invoked the IH on $\elab{\with \sigma', x \handle c}$.
\end{proof}

This elaboration is sound up to administrative reduction.
A step in \cambria{} forces a step by its elaboration, modulo the existential unpackings which carry no runtime content.
That is, each step is matched by some administrative steps from each elaboration, with a final principal step to close the gap.

\begin{proposition}[Operational Correspondence] \label{prop:op-corresp}
  Define a \emph{simulation relation} $c \lhd c'$ when the elaborations $c^\dagger$ and ${c'}^\dagger$ share a common reduct $d$,
  with $c^\dagger \corestep^*_a \; \corestep_p d$ and ${c'}^\dagger \corestep_a^* d$.
  Then $c \leadsto c'$ implies $c \lhd c'$.
\end{proposition}

\begin{proof}
  We first show by induction that $(\with \sigma, h \handle b)^\dagger \corestep^*_a (\with h^\dagger \handle b^\dagger)$.
  For $\sigma = \emptyset$, this is immediate after zero reductions.
  For $\sigma = \sigma'[\psi \mapsto A]$,
  \begin{align*}
    (\with \sigma'[\psi \mapsto A], h \handle b)^\dagger
    &= \unpack \pack(\psi, A, h^\dagger) \as (\psi, x) \iin (\with \sigma', x \handle b)^\dagger \\
    &\corestep_a (\with \sigma', h \handle b)^\dagger \\
    &\corestep_a^* \with h^\dagger \handle b^\dagger \tag{IH}
  \end{align*}
  For the reduction rules, the only non-trivial cases are handler reductions.
  Elaboration is structural identity on all other terms, so principal steps correspond exactly to \cambria{} steps, and $d = c'^\dagger$ witnesses $c \lhd c'$.
  For $c = \with \sigma, h \handle c_t$ where $c_t$ is a terminal form,
  \[
    c^\dagger \corestep^*_a \with h^\dagger \handle c_t^\dagger \corestep_p c'^\dagger
  \]
  because parametrized handling reinstates $\sigma = \emptyset$ and elaboration commutes with substitution.
  So $d = c'^\dagger$ witnesses $c \lhd c'$.
  
  Finally, we show that $c \lhd c'$ implies $E[c] \lhd E[c']$ by induction on the evaluation context.
  The base case is immediate, and the $\mathsf{do}$ context follows directly from the IH because elaboration is again structural identity.
  For $E = \with \sigma, h \handle E'$, assume $E'[c] \lhd E'[c']$ witnessed by $d$.
  Then,
  \begin{align*}
    E[c]^\dagger \corestep^*_a \with h^\dagger \handle E'[c]^\dagger \corestep^*_a \; \corestep_p \with h^\dagger \handle d, \\
    E[c']^\dagger \corestep^*_a \with h^\dagger \handle E'[c']^\dagger \corestep^*_a \with h^\dagger \handle d.
  \end{align*}
  So $(\with h^\dagger \handle d)$ witnesses $E[c] \lhd E[c']$.
\end{proof}

Since the elaboration produces only classic handlers and existential types,
and reduces to a classic handler when $\sigma=\emptyset$, a language offering these two features needs no new core constructs to extend support to parametrized handlers.

\subsection{Contextual Equivalence and Adequacy of the Logical Relation}

The operational correspondence allows us to state the logical relation over \cambria{} terms via their elaborations.
To reason with it, we connect the logical relation to the standard notion of program equivalence in \cambria{}, contextual equivalence.
We define two programs to be contextually equivalent when no context can distinguish their behaviour.
In the presence of general recursion, it is enough to observe termination, since any behavioural difference can be converted to divergence on one side.
Hence, two programs are contextually equivalent when, in any context, one terminates if and only if the other does.

\begin{definition}
  In \cambria{}, a \emph{computation context} is a computation with a hole in it.
  \[
  \begin{array}{rcl}
    \ctx & ::= & \square \mid \op(v; x.\ctx) \mid \ddo x \leftarrow \ctx \iin c
                                       \mid \ddo x \leftarrow c \iin \ctx
                                       \mid \with \sigma, v \handle \ctx
  \end{array}
  \]
  Given a well-typed computation $\cj \Theta \Psi \Gamma c A \Sig$,
  a \emph{closing context} for $c$ is a computation context $\ctx$ such that $\cj \eenv \eenv \eenv {\ctx[c]} \tunit \emptyset$.
  Given two well-typed computations $\cj \Theta \Psi \Gamma {c_1, c_2} A \Sig$, they are \emph{contextually equivalent}, written $c_1 \approx c_2$,
  if for all closing contexts $\ctx$,
  \[
    \ctx[c_1] \leadsto^* \ret \unit \iff \ctx[c_2] \leadsto^* \ret \unit.
  \]
\end{definition}

Adequacy states that the logical relation captures this notion of contextual equivalence in \cambria{}.
In particular, if any two elaborated, well-typed \cambria{} terms are in the logical relation for any relational interpretation, step index, and environment substitutions,
then they are contextually equivalent.
This is achieved via a plugging lemma which propagates the logical relation through any computation context,
bridging the fundamental theorem to contextual equivalence.

\begin{lemma}[Plugging Lemma] \label{lem:plug}
  If for all $\rho: \Theta \cup \Psi \to \Rel$, $n\in\mathbb{N}$, and $(\gamma_1, \gamma_2)_n \in \interp{\Gamma}$, we have
  \[
    (c_1^\dagger[\gamma_1], c_2^\dagger[\gamma_2])_n \in \interp{\cC},
  \]
  then for any computation context $\ctx$, where $\cjp {\Theta'} {\Psi'} {\Gamma'} {\ctx[c_1], \ctx[c_2]} {\cC'}$, we have for all 
  $\rho': \Theta' \cup \Psi' \to \Rel$, $n\in\mathbb{N}$, and $(\gamma_1', \gamma_2')_n \in \interp[\rho']{\Gamma'}$,
  \[
    (\ctx[c_1]^\dagger[\gamma_1'], \ctx[c_2]^\dagger[\gamma_2'])_n \in \interp[\rho']{\cC'}.
  \]
\end{lemma}

\begin{proof}
  The proof is by structural induction on the computation context.
  In each case, the fundamental theorem (\Cref{thm:ft})
  relates non-hole subterms to themselves, and the IH relates the subcontext.
  The only non-standard case is $(\with \sigma, v \handle \ctx')$, which uses inner induction on $\sigma$.
  The base case proceeds with the same pattern (relating non-hole terms via \Cref{thm:ft} and subcontexts via the IH).
  The inductive case elaborates as
  \[
    (\with \sigma'[\psi \mapsto A], v \handle \ctx')^\dagger
    = \unpack \pack(\psi, A, v^\dagger) \as (\psi, x) \iin (\with \sigma', x \handle \ctx')^\dagger,
  \]
  which reduces back to $(\with \sigma', v \handle \ctx')^\dagger$, and we use the inner IH to relate the two terms after environment substitutions.
\end{proof}

The logical relation is asymmetric because the step index tracks reductions on one side, capturing step-bounded approximate behaviour.
The approximation becomes a full bidirectional equivalence for a relational interpretation $\rho$ when for all $n\in \mathbb{N}$,
$(x_1, x_2)_n \in \interp{X}$ and $(x_2, x_1)_n \in \interp[\rho^{-1}]{X}$, where $\rho^{-1}$ is the relational interpretation inverted pointwise.
In this case we write $(x_1, x_2) \in \siminterp{X}$.

\begin{theorem}[Adequacy] \label{thm:adequacy}
  Let $\cjp \Theta \Psi \Gamma {c_1, c_2} \cC$. Suppose for all relational interpretations $\rho : \Theta \cup \Psi \to \Rel$
  and related environments $(\gamma_1, \gamma_2) \in \siminterp{\Gamma}$, we have $(c_1^\dagger[\gamma_1], c_2^\dagger[\gamma_2]) \in \siminterp{\cC}$.
  Then $c_1 \approx c_2$.
\end{theorem}

\begin{proof}
  Let $\ctx$ be any context where $\cjp \eenv \eenv \eenv {\ctx[c_1], \ctx[c_2]} {\comp{\tunit}{\emptyset}}$. By \Cref{lem:plug} applied in both directions,
  \begin{equation}
    (\ctx[c_1]^\dagger, \ctx[c_2]^\dagger) \in \siminterp[{}]{\comp{\tunit}{\emptyset}}. \label{eq:plugged}
  \end{equation}
  Hence, by definition of the logical relation, $\ctx[c_1]^\dagger$ terminates at $\ret\unit$ if and only if $\ctx[c_2]^\dagger$ does,
  because the operation signature is empty and the only elements of $\interp[{}]{\tunit}$ are $(\unit, \unit)_n$.
  Next, we show for $\cj \eenv \eenv \eenv c \tunit \emptyset$,
  \begin{equation}
    c \leadsto^* \ret\unit \iff c^\dagger \corestep^* \ret\unit. \label{eq:termunit}
  \end{equation}
  The forward direction is given by repeated application of \Cref{prop:op-corresp} using that core reductions are deterministic.
  In the reverse direction, each \cambria{} step induces at least one core step (\Cref{prop:op-corresp}),
  so the finite deterministic core trace implies a finite \cambria{} trace.
  By \Cref{cor:safety}, the \cambria{} reduction terminates at $\ret\unit$, the only terminal form of type $\comp{\tunit}{\emptyset}$.
  Thus from \eqref{eq:plugged} and \eqref{eq:termunit}, $c_1 \approx c_2$.
\end{proof}

\subsection{Applications of Parametricity} \label{sec:applications}

Beyond the standard \emph{free theorems}~\cite{free-theorems} provided by parametricity results,
two consequences are key in the parametrized handler setting:
representation independence between parametrized handler implementations and the ability to reason equationally within the parametrized effect theories.
We present each as an implication of \Cref{thm:ft}, illustrating them with the examples from \Cref{sec:local-state}.

\subsubsection{Representation Independence for Parametrized Handlers} \label{sec:repind}

Two implementations of the same parametrized effect should be indistinguishable to client code that only interacts through its abstract interface.
Parametricity guarantees this.
If any two handlers agree up to an admissible relation on the parameter, the handled computations are observationally equivalent.

\begin{corollary} \label{thm:rep-ind}
  Let $\sigma_1 = \{\psi \mapsto A_1\}$ and $\sigma_2 = \{\psi \mapsto A_2\}$ be two well-typed parameter substitutions, and let
  \[
    \cjp{\Theta}{\Psi, \psi}{\Gamma}{c}{\cC}, \qquad
    \vj{\Theta}{\Psi}{\Gamma}{h_1}{\thandler{\cC[\sigma_1]}{\cD}}, \qquad
    \vj{\Theta}{\Psi}{\Gamma}{h_2}{\thandler{\cC[\sigma_2]}{\cD}}.
  \]
  Suppose that for every relational interpretation $\rho : \Theta \cup \Psi \to \Rel$ and related environment
  $(\gamma_1, \gamma_2) \in \siminterp{\Gamma}$, there exists an admissible relation $R \in \Rel$ such that
  \begin{equation}\label{eq:repind-premise}
    (h_1^\dagger[\gamma_1], h_2^\dagger[\gamma_2]) \in \siminterp[{\rho[\psi \mapsto R]}]{\thandler{\cC}{\cD}}.
  \end{equation}
  Then $(\with \sigma_1, h_1 \handle c) \approx (\with \sigma_2, h_2 \handle c)$.
\end{corollary}

\begin{proof}
  Let $\rho : \Theta \cup \Psi \to \Rel$ and $(\gamma_1, \gamma_2) \in \siminterp{\Gamma}$.
  By \Cref{thm:adequacy}, we require
  \begin{equation}\label{eq:repind-goal}
    ((\with \sigma_1, h_1 \handle c)^\dagger[\gamma_1],\; (\with \sigma_2, h_2 \handle c)^\dagger[\gamma_2]) \in \siminterp{\cD}.
  \end{equation}
  Let $R$ be the witness for \eqref{eq:repind-premise}.
  Since $\psi \notin \fpv(\Gamma)$, $(\gamma_1, \gamma_2) \in \siminterp[{\rho[\psi \mapsto R]}]{\Gamma}$.
  By \Cref{thm:ft} applied in both directions, with the inverse relational interpretation for the reverse,
  \[
    (c^\dagger[\gamma_1], c^\dagger[\gamma_2]) \in \siminterp[{\rho[\psi \mapsto R]}]{\cC}.
  \]
  The handler clause of the logical relation gives
  \[
    (\with h_1^\dagger[\gamma_1] \handle c^\dagger[\gamma_1], \with h_2^\dagger[\gamma_2] \handle c^\dagger[\gamma_2]) \in \siminterp[{\rho[\psi \mapsto R]}]{\cD}
    = \siminterp{\cD}
  \]
  as $\psi \notin \fpv(\cD)$. Finally,
  \begin{align*}
    (\with \sigma_i, h_i \handle c)^\dagger[\gamma_i]
      &= \unpack \pack(\psi, A_i, h_i^\dagger[\gamma_i]) \as (\psi, x) \iin (\with x \handle c^\dagger[\gamma_i]) \\
      &\corestep_a (\with h_i^\dagger[\gamma_i] \handle c^\dagger[\gamma_i]).
  \end{align*}
  So by anti-reduction, \eqref{eq:repind-goal} holds.
\end{proof}

This result extends to substitutions of multiple parameters $\sigma = \{\psi_1 \mapsto A_1, \ldots, \psi_k \mapsto A_k\}$ by induction on $k$.
The elaboration unfolds into $k$ nested unpack-pack layers, each contributing one $\corestep_a$ step. The proof is otherwise unchanged.

\begin{example} \label{ex:fresh-disc}
  We sketch a use case of \Cref{thm:rep-ind} for two implementations of the fresh effect.
  Take $h_1$ to be the fresh effect handler from \Cref{sec:local-state} with $\sigma_1 = \{\psi \mapsto \tint\}$ and
  take $h_2$ to be a string-counter variant with $\sigma_2 = \{\psi \mapsto \tstring\}$:
\begin{lstlisting}
handler {
  return x   -> return (fun _ -> return x),
  fresh _ k  -> return (fun s -> k s ("x" ++ s)),
  eq (a,b) k -> k (a == b),
  finally f  -> f ""
}
\end{lstlisting}
  producing the sequence of strings $\mlstinline{""}, \mlstinline{"x"}, \mlstinline{"xx"},\dots$ for the names.

  We choose the admissible relation $R = \{(i, \xstr^i)_n \mid i,n\in \mathbb{N}\}$, where $\xstr^i$ is the string of $i$ copies of \lstinline|x|.
  The clauses of the handlers satisfy this choice of $R$ and its inverse because:
  \begin{itemize}
    \item \lstinline|return| each produces a function ignoring the parameter, so the results are related.
    \item \lstinline|fresh| applied to related $i$ and $\xstr^i$ invokes the continuation on this pair,
      and the resulting function is applied to related $i{+}1$ and $\xstr^{i{+}1}$.
    \item \lstinline|eq| applied to related arguments $(i, j)$ and $(\xstr^i, \xstr^j)$ gives related booleans because $i \leftrightarrow \xstr^i$ is injective in both directions.
  \end{itemize}

  So the two implementations agree up to $R$, satisfying the premise of \Cref{thm:rep-ind}.
  Hence, any body that uses the fresh effect is observationally equivalent under $h_1$ and $h_2$.
\end{example}

The relation may do work beyond bijective encodings.
The two implementations of P\'olya's urn from \Cref{sec:polya} are such an example.
The first stores the urn's state as a pair of integers in a memory location while the other uses a single double,
with the witnessing $R$ a \emph{coupling} between a cell holding state $(r,b)$ and a draw $p$ from the Beta distribution $\Beta(r,b)$.
The deterministic logical relation already accounts for half of the equivalence,
where the abstraction over \lstinline|$urn| ensures that no client can observe which representation a handler uses.
But it does not establish that the two handlers satisfy $R$ because \lstinline|!uniform| and \lstinline|!bernoulli| are not part of the core calculus (\Cref{sec:calculus}).
The handlers sample differently, only agreeing in distribution by conjugacy of the Beta prior to Bernoulli observations~\cite{staton2018beta}.
That is, $\mathbb{E}[\Beta(r,b)] = r/(r+b)$ matches the predictive probability of red,
and conditioning on a draw updates it to $\Beta(r+1,b)$ or $\Beta(r,b+1)$, matching the urn update.
So, parametricity reduces the proof obligation to the conjugacy coupling and the implementation of the probabilistic primitives.
Making this formal requires lifting the logical relation along a probability monad,
replacing membership in an admissible relation with a coupling lifting~\cite{gregersen2024clutch}
and re-establishing anti-reduction and compositionality for the lifted relation (e.g. via entropy-passing semantics~\cite{culpepper2017contextual, wand2018contextual}),
which we leave to future work.

In contrast, the three dynamic thread handlers from \Cref{sec:concurrency} are not representation independent of one another
because they differ in their scheduling strategies.
The child-biased and parent-biased handlers \emph{choose} different orderings for the threads,
so while each is a sound interpretation of the operations, they are not observationally equivalent.
\Cref{thm:rep-ind} can still be used to relate handlers with the same scheduling strategy differing in the representation of \lstinline|$tid|.
For example, a child-biased handler that instead passes the parent thread an arbitrary integer with instantiation \lstinline|[$tid -> Int]|
is observationally equivalent to the child-biased handler of \Cref{sec:concurrency} that uses \lstinline|Unit|,
since the client cannot inspect the thread ID.
\begin{lstlisting}
[$tid -> Int] handler {
  ...
  fork _ k -> k (inr ()) ; k (inl 57)
}
\end{lstlisting}

\subsubsection{Equational Properties of Parametrized Handlers} \label{sec:equations}

A standard use of effect handlers is to verify implementations against the equational laws of the effect theory.
With local state, initializing a reference then immediately overwriting it should be equivalent to initializing it to the latter value.
\[
  \rref(v;\; a.\sset((a,w);\; \unit.c)) \; \approx \; \rref(w;\; a.c)
\]
Under a function-based handler, both sides reduce to a function with $a$ pointing to $w$, and the equivalence follows.

Equations where a different number of resources are allocated on each side are subtler.
The \emph{garbage collection law} states that allocating a reference and never using it is equivalent to having never allocated it.
\[
  \rref(v;\; a.c) \; \approx \; c \tag{$a \notin \fv(c)$}
\]
Under the function-based handler, the two sides produce functions that differ at $a$. A proof would have to show that this entry is unobservable,
and parametricity provides exactly this argument. We demonstrate the law for the fresh effect.

\begin{example} \label{ex:gc}
  Let $\cj \Theta {\Psi, \psi} \Gamma c A \Sig$, where $\Sig \supseteq \{ \fresh : \arity{\tunit}{\psi}, \eq : \arity{\psi \times \psi}{\tbool} \}$,
  $\psi \notin \fpv(A, \Gamma)$, and $a \notin \fv(c)$. Then
  \[
    \with \{\psi \mapsto \tint\}, h_{\nu} \handle c \; \approx \; \with \{\psi \mapsto \tint\}, h_{\nu} \handle \fresh(\unit; \; a.c),
  \]
  where $h_\nu$ is the integer parameter passing fresh effect handler from \Cref{sec:local-state}.
\end{example}

\begin{proof}
  Elaborating each side, the existential layer reduces to a plain handle, and the \lstinline|finally| clause of
  $h_\nu$ feeds the handler's output $f$ to the application $f\;0$.
  On the right, the $\fresh$ command is handled, incrementing the counter by one. As $a \notin \fv(c)$,
  each side respectively reduces to
  \[
    c_1 = \ddo f \leftarrow (\with \nofin{h}_{\nu}^\dagger \handle c^\dagger) \iin f \; 0,
    \qquad
    c_2 = \ddo f \leftarrow (\with \nofin{h}_{\nu}^\dagger \handle c^\dagger) \iin f\; 1,
  \]
  where $\nofin{h}_{\nu}$ is the handler without the \lstinline|finally| clause, typed at
  $\thandler{\comp{A}{\Sig}}{\comp{(\tfun{\psi}{\comp{A}{\Sig'}})}{\Sig'}}$ with $\Sig' = \Sig \setminus \{\fresh, \eq\}$.
  So the problem reduces to showing $c_1 \approx c_2$.

  We let $\rho : \Theta \cup \Psi \to \Rel$ and choose the admissible relation $R = \{(i, i{+}1)_n \mid i, n \in \mathbb{N}\}$.
  This choice of $R$ maintains the relatedness of the handlers much like in \Cref{ex:fresh-disc}.
  The \lstinline|return| clause produces a function constant in the parameter argument,
  \lstinline|fresh| maps the related pair $(i, i{+}1)$ to the related pair $(i{+}1, i{+}2)$,
  and \lstinline|eq| is invariant because $+1$ is injective in both directions.
  So $(\nofin{h}_{\nu}^\dagger,\, \nofin{h}_{\nu}^\dagger) \in
  \siminterp[{\rho[\psi\mapsto R]}]{\thandler{\comp{A}{\Sig}}{\comp{(\tfun{\psi}{\comp{A}{\Sig'}})}{\Sig'}} }$.

  Let $(\gamma_1, \gamma_2) \in \siminterp{\Gamma} = \siminterp[{\rho[\psi\mapsto R]}]{\Gamma}$,
  since $\psi \notin \fpv(\Gamma)$. By the fundamental theorem (\Cref{thm:ft}) applied in both directions (inverting $\rho[\psi\mapsto R]$ for the reverse),
  \[
    \bigl((\with \nofin{h}_\nu^\dagger \handle c^\dagger)[\gamma_1],\;
    (\with \nofin{h}_\nu^\dagger \handle c^\dagger)[\gamma_2]\bigr)
    \in \siminterp[{\rho[\psi\mapsto R]}]{ \comp{(\tfun{\psi}{\comp A {\Sig'}})}{\Sig'} }.
  \]
  Since $(0, 1) \in R$ at every step index, the function relation and bind lemma (\Cref{lem:bind}) give
  \[
    (c_1[\gamma_1], c_2[\gamma_2]) \in \siminterp[{\rho[\psi\mapsto R]}]{\comp{A}{\Sig'}} = \siminterp{\comp{A}{\Sig'}},
  \]
  since $\psi \notin \fpv(A, \Sig')$. Then anti-reduction lifts this to the original computations, and adequacy to contextual equivalence in \cambria{}.
\end{proof}

\paragraph{Parametricity and Reachability Analysis}
\Cref{ex:gc} states that allocating an unused reference is observationally undetectable.
However, the function-based handler of \Cref{sec:local-state} retains the entry for that reference forever, since nothing signals when a reference has gone out of scope.
While observationally harmless, it is a form of memory leak because the value persists throughout the entire body of the handle.
This behaviour is unavoidable.
A standard free theorem implies that naive liveness detection and parametricity are incompatible.

\begin{corollary}[Vacuousness~\cite{free-theorems}] \label{cor:vacuous}
  Any closed \cambria{} function of type $\forall \alpha. (\tfun{\alpha}{\comp A \emptyset})$ where $\alpha$ is not free in $A$
  is extensionally constant. It returns the same value (or diverges) for every input.
\end{corollary}

Any garbage-collection scheme that determines liveness from references occurring in the captured continuation requires
a primitive that accepts the continuation as an arbitrary value (to work uniformly across handler instantiations),
and produces an output with type independent of the polymorphic input.
For example, the following are the type signatures for the primitives \lstinline|isReachable|,
which computes whether a reference \lstinline|Ref| is still reachable in a continuation,
\lstinline|liveRefs|, which finds the live references in a continuation,
and \lstinline|onUnreachable|, which registers a finalizer to fire when a reference becomes unreachable.
\begin{lstlisting}
  isReachable   : foralla. a -> (Ref -> Bool ! {}) ! {}
  liveRefs      : foralla. a -> List Ref ! {}
  onUnreachable : foralla. a -> (Ref * (Unit -> Unit ! {}) -> Unit ! {}) ! {}
\end{lstlisting}
By \Cref{cor:vacuous}, such primitives cannot perform non-trivial liveness analysis.

Reachability-based garbage collection for parametric local state is thus beyond the scope of this paper.
Substructural typing, regions, and built-in reference effects remain possible alternatives, which we discuss in \Cref{sec:conclusion}.

\section{Type Inference}\label{sec:inference}

So far we have established \cambria{}'s type safety and parametricity.
We now show that parametrized handlers are also practical, introducing only a minimal annotation burden.
Since parameters abstract concrete types within the body of a handle, annotations are necessary to indicate which types are intended to be abstracted.
The main result of this section, \Cref{thm:relevant-complete}, is that it suffices to annotate only the operation arities that introduce new parameters.
The type checker can infer everything else.

Inference for \cambria{} follows Algorithm W for Hindley--Milner type systems~\cite{milner_theory_1978, damas_principal_1982},
extended to effects~\cite{pretnar2014inferring}.
The inference algorithm uses \emph{row variables}~\cite{leijen_koka_2014, lindley_row-based_2012},
which may appear at the end of effect signatures and are required to encode operations not yet observed at that program point.
The algorithm traverses the syntax tree, unifying value types and effect rows eagerly, and generalizing at each \textsc{Do} binding.
There is no value restriction at generalization, which is sound in this setting~\cite{kammar_no_2017},
because only variables not free in the environment or effect row are generalized, and recursive functions are monomorphic.
An operation invocation $!\op\;v$ adds $\op$ to the effect row with fresh arity, unifying the input with the type of $v$ and the output with the type of $!\op\;v$.
Using an effect declaration (\Cref{sec:surface}) unifies the declared arity with the operation in the effect row, if it exists.
Handler clauses are checked as in the declarative \textsc{Handler} rule, with the shared answer type and output row found by unification.
The forwarding side condition $\Sig\setminus\Sig''\subseteq\Sig'$ is discharged by unifying the output row with the input row variable and unhandled operations,
rather than a separate subset test.

The only construct that departs from established inference is \textsc{ParHandle}.
Parameters only enter the body through type or effect declarations.
Unification, denoted $\doteq$, treats each parameter as an additional rigid constant, distinguishing them from unification variables.
That is, a parameter only unifies with itself and fails against any other type, including a distinct parameter.
At a handle term, the algorithm applies $\sigma$ to the inferred body type and unifies the result against the handler's input type.
The no-escape condition $\dom(\sigma) \cap \fpv(\cD, \G) = \emptyset$ of \textsc{ParHandle}
is an additional side condition ensuring no parameter has leaked into the context or result type via unification.
We write $\wj \Theta \Psi \G t X$ to indicate the algorithm infers $X$ for $t$.
\begin{mathpar}
  \inferrule*[right=ParHandle]{
    \wj \Theta \Psi \G v {\thandler{\cC_h}{\cD}}
    \\
    \wj \Theta {\Psi,\dom(\sigma)} \G c {\cC_c}
    \\
    [\tj \Theta \Psi A]_{A \in \cod(\sigma)}
    \\\\
    \cC_c[\sigma] \doteq \cC_h
    \\
    \dom(\sigma) \cap \fpv(\cD, \G) = \emptyset
  }{
    \wj \Theta \Psi \G {\with \sigma, v \handle c} \cD
  }
\end{mathpar}
This is the new machinery in total.
The remaining inference is ordinary Algorithm W with the parameters treated as constants.
The algorithm is syntax-directed and terminating, hence decidable, and unifying value types and effect rows has principal solutions,
being first-order unification~\cite{damas_principal_1982} and row unification~\cite{leijen_koka_2014, lindley_row-based_2012} with parameters as constants.

We establish \emph{principality} of \cambria{} by showing soundness and completeness of the inference algorithm.
Soundness states that every type inferred by the algorithm is derivable in the declarative system.
Hence, the algorithm accepts only well-typed programs that have the type-safety and parametricity guarantees.
Completeness, treated next, characterizes the annotations for which the algorithm infers the most general derivable type for a computation.

\begin{theorem}[Soundness] \label{thm:type-soundness}
  If $\wj{\Theta}{\Psi}{\G}{c}{\cC}$, then $\cjp{\Theta}{\Psi}{\G}{c}{\cC}$.
\end{theorem}

\begin{proof}
  \cambria{} type inference differs from Hindley--Milner with effect handlers only in the handle rule generalized to the \textsc{ParHandle} rule.
  Since parameters are rigid, unification remains sound, and all remaining cases follow from established proofs~\cite{pretnar2014inferring, damas_principal_1982, kammar_no_2017}.

  For \textsc{ParHandle} ($\with \sigma, v \handle c$), the IH on $v$ gives a derivable handler type $\thandler{\cC[\sigma]}{\cD}$.
  The IH on $c$ gives a derivable body type $\cC$, with $\dom(\sigma)$ rigid in $\Psi$.
  Unifying $\cC[\sigma]$ against the handler's input restricts them to be equal, where the types in $\cod(\sigma)$ are well-formed.
  For the no-escape condition, the algorithm explicitly checks that $\dom(\sigma) \cap \fpv(\cD, \G) = \emptyset$,
  so $\cjp{\Theta}{\Psi}{\G}{\with \sigma, v \handle c}{\cD}$ is derivable.
\end{proof}

We write $\cC \sqsubseteq \cC'$ when $\cC$ arises from $\cC'$ by substituting value types for its type variables and \emph{closed} signatures for its row variables.
Closing the row variables recovers the closed-row typing of \Cref{sec:calculus}.
We develop completeness in three stages.
The first states that any term that does not introduce a parameter is inferable without any annotations because inference is exactly Algorithm W with effect handlers.

\begin{lemma}[Ground Completeness] \label{thm:ground-complete}
  Let every handle expression in $c$ have $\sigma = \emptyset$.
  If $\cjp{\Theta}{\Psi}{\G}{c}{\cC}$, then $\wj{\Theta}{\Psi}{\G}{c}{\cC'}$ with $\cC \sqsubseteq \cC'$.
\end{lemma}

\begin{proof}
  Inference over this fragment is identical to Algorithm W with effects,
  for which completeness has been established~\cite{damas_principal_1982, pretnar2014inferring}.
\end{proof}

In the full parametrized system, an \emph{annotated term} is one in which every handle expression carries an explicit type annotation on the body.
\[
  \with \sigma, v \handle (c : \cC)
\]

\begin{lemma}[Parametrized Completeness] \label{thm:param-complete}
  Let $c$ be an annotated term. If $\cjp{\Theta}{\Psi}{\G}{c}{\cC}$, then $\wj{\Theta}{\Psi}{\G}{c}{\cC'}$ with $\cC \sqsubseteq \cC'$.
\end{lemma}

\begin{proof}
  By structural induction on the typing derivation.
  All cases apart from \textsc{ParHandle} follow the Hindley--Milner argument for algebraic effects~\cite{pretnar2014inferring, damas_principal_1982}.
  Since parameters are rigid and never generalized, unification and generalization arguments are unaffected.
  In particular, gen-admissibility~\cite{damas_principal_1982} still holds, so \textsc{Do} is unchanged.

  The only interesting case is \textsc{ParHandle}: $\with \sigma, v \handle (c : \cC)$.
  By the IH, the algorithm infers a handler type $\thandler{\cC'}{\cD'}$ such that $\thandler{\cC[\sigma]}{\cD} \sqsubseteq \thandler{\cC'}{\cD'}$.
  The body is then checked against the annotation $\cC$ with $\dom(\sigma)$ rigid in $\Psi$,
  and the algorithm unifies $\cC[\sigma] \doteq \cC'$.
  Since unification returns the most general solution, the final inferred output type $\cD''$ satisfies $\cD \sqsubseteq \cD''$.
\end{proof}

The annotation in \Cref{thm:param-complete} uses the full type $\cC$ of the body.
We strengthen this to a \emph{relevant} annotation consisting only of the operations that mention the introduced parameters.
This result relies on operation arities being monomorphic (\Cref{sec:calculus}).
Additional annotations would be required for universally quantified arities, and we leave the interplay of polymorphic and parameter annotations to future work.

\begin{definition} \label{def:relevant}
  Given a parameter substitution $\sigma$, the \emph{$\sigma$-relevant operations} of $\Sig$ are
  \[
    \rel_\sigma(\Sig)
    \;=\;
    \bigl\{\,\op : A_\op \to B_\op \in \Sig \mid
      \dom(\sigma) \cap (\fpv(A_\op) \cup \fpv(B_\op)) \neq \emptyset
    \,\bigr\}.
  \]
  A \emph{relevantly annotated term} is one in which each handle expression carries an annotation listing only the $\sigma$-relevant operations,
  $\with \sigma, v \handle (c : \rel_\sigma(\Sig))$. The return type and all other operations are omitted.
  Note that this partial annotation is possible via effect declarations (\Cref{sec:surface}).
\end{definition}

\begin{theorem}[Relevant Completeness] \label{thm:relevant-complete}
  Let $c$ be a relevantly annotated term.
  If $\cjp{\Theta}{\Psi}{\G}{c}{\cC}$, then $\wj{\Theta}{\Psi}{\G}{c}{\cC'}$ with $\cC \sqsubseteq \cC'$.
\end{theorem}

\begin{proof}
  We extend the proof of \Cref{thm:param-complete},
  strengthening the \textsc{ParHandle} case to show that from the relevant annotations, the algorithm infers a type at least as general as the body's declarative type.
  The handle introduces $\dom(\sigma)$ into $\Psi$ as rigid constants, and the operations split into two kinds:
  \begin{enumerate}
    \item \emph{$\sigma$-relevant.}
      These arities are supplied by the annotation, so they are typed exactly as in \Cref{thm:param-complete}.
    \item \emph{$\sigma$-free.}
      For any such $\op : \arity{A_\op}{B_\op}$, we have $\dom(\sigma) \cap (\fpv(A_\op) \cup \fpv(B_\op)) = \emptyset$,
      so $\sigma$ is the identity on the arity, and the op is inferred as in the unparametrized system with no annotation.
  \end{enumerate}
  \emph{Return type.}
  A value or type mentioning $\psi\in\dom(\sigma)$ can enter the body only through a $\sigma$-relevant operation.
  Hence, once the $\sigma$-relevant operations are annotated, every $\psi$ is fixed wherever it flows.
  The return type of the body is then inferred with the parameters acting as fixed base types.
\end{proof}

Thus, the annotation burden is minimal.
The local-state example of \Cref{sec:local-state} is annotated only with
the three $\sigma$-relevant arities \lstinline|!get : $loc ~> Int|, \lstinline|!set : $loc * Int ~> Unit|, and \lstinline|!ref : Int ~> $loc|,
which is enough to infer the type of the whole program.
Moreover, the relevant annotation is a term-independent sufficient condition.
Operations in $\rel_\sigma(\Sig)$ may be omitted from annotations when $\psi$ propagates to their arity via term-level flow.
For example, annotating only \lstinline|!set| in
\begin{lstlisting}
  effect !set : $loc * Int ~> Unit.
  do x <- !ref 3 in !set (x, 4)
\end{lstlisting}
permits the algorithm to infer the arity of \lstinline|!ref| is \lstinline|Int ~> $loc|.
For this reason, \cambria{} enforces no constraints on which operations need to be annotated,
letting the programmer manage the interface with individual effect declarations.

\section{Related Work}\label{sec:related}

\paragraph{Algebraic effects and handlers.}
Algebraic effects were introduced by Plotkin and Power~\cite{plotkin_algebraic_2003} and equipped with handlers by Plotkin and Pretnar~\cite{plotkin_handling_2009}.
They are now incorporated in numerous languages, such as Eff~\cite{bauer_programming_2015}, Koka~\cite{leijen_koka_2017},
Frank~\cite{lindley_do_2017}, Links~\cite{hillerstrom2016liberating}, and OCaml~5~\cite{sivaramakrishnan_retrofitting_2021}.
Handlers in OCaml~5 are widely deployed, but the language does not track effect use in its type system; the focus is rather on handlers as a control mechanism.
\cambria{} is closest in spirit to the fine-grained call-by-value calculus of Kammar et al.~\cite{kammar_no_2017},
whose type-safety results we extend to the parametrized setting.

\paragraph{Parametrized algebraic theories.}
\cambria{} is inspired by \emph{parametrized algebraic theories}~\cite{Staton13, Staton13PL} (PATs),
where operations may bind and use abstract parameters, giving an algebraic account of effects that allocate resources.
PATs have been used to axiomatize a range of effects, including local state~\cite{Staton13}, code jumps~\cite{staton2014jumps},
the Beta-Bernoulli process~\cite{staton2018beta}, quantum measurement~\cite{staton-quantum}, and scoped effects~\cite{matache-scoped}.
\cambria{} complements this line of work by realizing it in the algebraic effects and handlers framework, with parametricity enforcing the abstraction layer.
Concretely, an operation returning a parameter, such as \lstinline|!fresh : Unit ~> $name|,
is the \emph{generic effect}~\cite{plotkin_algebraic_2003} of a PAT operation that binds that parameter in its continuation.
Another recent realization, Paella~\cite{sigal2024paella}, is an Idris~2 library that models dynamically allocated resources through a Kripke possible-world semantics,
where types are indexed by the resources currently allocated.
\cambria{} addresses the same setting as a stand-alone calculus where parameters are abstract types erased at runtime,
and it establishes parametricity rather than providing a semantic framework.

\paragraph{Effect instances and named handlers.}
The \emph{problem of instances}, where effects dynamically allocate resources, has motivated a substantial line of work.
Eff originally had first-class effect instances~\cite{bauer_programming_2015}, but removed them because the dynamic
correspondence between operations and handlers complicated reasoning.
Subsequent calculi recovered instances by a range of methods,
including dynamically allocated labels~\cite{biernacki2020binders, devilhena2023labels} and capability passing~\cite{brachthauser2020capabilities}.
As a practical language, Koka~\cite{leijen_koka_2014} uses \emph{named handlers}~\cite{xie2022firstclass, xie2020evidently},
which expose each instance as a first-class name passed as evidence at runtime.
This can model dynamically allocated heaps with first-class references, and rank-2 polymorphism scopes the instances to prevent leakage.
In each of the approaches above, the instance is a concrete object in the program.
\cambria{} takes a different approach, separating dispatch from instances.
A handler interprets all instances of an effect, and the resource is exposed only as an abstract \emph{parameter} that the handler instantiates.
Thus, \cambria{} guarantees \emph{representation independence} over the user-defined resource.
Moreover, there is no type-directed reduction, and since a parameter is a genuine type rather than a witness bound to an instance,
a single handler can coordinate distinct instances, as in the dynamic threads (\Cref{sec:concurrency}) and jump~\cite{staton2014jumps} effects.

\paragraph{Parametricity and logical relations.}
Parametricity~\cite{reynolds1983types} is the source of \cambria{}'s abstraction guarantees,
stemming from reading parameters as existential types~\cite{mitchell1988existential}
and data abstraction as representation independence~\cite{mitchell1986representation}.
We establish this through a step-indexed logical relation~\cite{appel_indexed_2001},
which was extended to recursive and quantified types by Ahmed~\cite{ahmed_semantics_2006}, and which we adapt to handlers.
This has been done biorthogonally~\cite{biernacki2018handle} by coupling the computation relation to the evaluation contexts,
and abstraction over effect signatures has separately been obtained through a type discipline~\cite{biernacki2019abstracting}.
We move the abstraction to resource identifiers and take the direct approach, where each type contributes an independent clause and new features are local.

\paragraph{Reasoning about dynamic allocation.}
Relational models with dynamic allocation often require Kripke logical relations indexed by possible worlds,
such as in generative, state-dependent ADTs~\cite{ahmed2009state}, or fresh-name generation in the $\nu$-calculus~\cite{zhang2003logical}.
\cambria{} has similar guarantees for local state and fresh-name effects (\Cref{sec:local-state}) \emph{without} appealing to possible worlds.
The relational interpretation of the parameter type plays the role that worlds play for a primitive full mutable heap.
Complementary to our type-based guarantees are program logics for handlers, such as the Hazel separation logic~\cite{devilhena2021separation}
and its relational counterpart Blaze~\cite{devilhena2026relational},
which, on a per-program basis, verify properties of handler programs over real mutable state and, for Blaze, concurrency.

\paragraph{Concurrent effects.}
Kammar et al.~\cite{kammar2025equational} develop an equational theory for dynamic threads as an algebraic effect,
using concrete thread IDs without type abstraction.
\cambria{}'s parameter variables make thread IDs abstract, giving parametricity guarantees and idiomatic handlers (\Cref{sec:concurrency}) for concurrent programs.

\section{Conclusion}\label{sec:conclusion}

We have presented \cambria{}, a calculus that extends algebraic effects and handlers with parameters to provide type abstraction for resource-allocating effects.
Parameter variables are first-class in the type system but play no role at runtime.
A step-indexed logical relation establishes parametricity, where client code cannot inspect the concrete representation of the parameter outside of the effect interfaces.
The decomposition of parametrized handlers into ordinary handlers and existential types
means that the surface language stays ergonomic while the metatheory reduces to standard, modular pieces.

We have demonstrated \cambria{}'s expressiveness with examples including local state, P\'olya's urn, and concurrent thread management.
The calculus is type safe and backwards compatible,
and type inference stays practical, as demonstrated in the implementation, with only effect arities mentioning parameters needing annotations.
These results are backed by a mechanization in the Rocq prover.
In the dynamic threads example, the thread IDs go beyond standard instances by coordinating concurrent processes.
Yet they remain abstract, providing a representation-independence guarantee.

\paragraph{Future work.}
Natural extensions include substructural parameters to incorporate quantum effects~\cite{staton-quantum} and scoped effects~\cite{matache-scoped}.
Such parameters could also address the reachability limitation of \Cref{sec:equations},
alongside built-in references and region-based memory management~\cite{tofte1997region,lorenzen2024oxidizing}.
A probabilistic semantics would also allow internalization of de Finetti-style representation equivalences like the one between the two P\'olya urn handlers (\Cref{sec:polya}).
We are also interested in whether parameter abstraction offers an alternative to tunnelling~\cite{zhang2019tunneling},
preventing effect-polymorphic code from accidentally handling effects it is unaware of.
Other directions include refining the type system to be bidirectional~\cite{dunfield2021bidirectional, lindley_do_2017},
and a denotational semantics to establish the categorical structure of parametrized handlers.

\bibliographystyle{ACM-Reference-Format}
\bibliography{references}

\end{document}